\documentclass[a4paper,10pt]{article}

\usepackage[margin=0.92in]{geometry} 

\usepackage{authblk}

\usepackage{amsmath}
\usepackage{amsthm}
\usepackage{amssymb}
\usepackage{soul}

\usepackage[utf8]{inputenc}
\usepackage{hyperref}
\hypersetup{
	unicode,
	pdfauthor={Jean Gasnier, Virgile Guemard},
	pdftitle={},
	pdfsubject={},
	pdfproducer={LaTeX},
	pdfcreator={pdflatex}
}

\usepackage[shortlabels]{enumitem}

\usepackage[sort&compress,numbers,square]{natbib}

\theoremstyle{plain}
\newtheorem{theorem}{Theorem}[section]
\newtheorem{corollary}[theorem]{Corollary}
\newtheorem{lemma}[theorem]{Lemma}

\newtheorem{proposition}[theorem]{Proposition}

\theoremstyle{definition}
\newtheorem{definition}[theorem]{Definition}

\newtheorem{remark}[theorem]{Remark}

\usepackage{graphicx, color}
\graphicspath{{fig/}}

\usepackage{algorithm, algpseudocode} 
\usepackage{mathrsfs}

\usepackage{dcolumn}
\usepackage{bm}
\usepackage[english]{babel}
\usepackage[OT1]{fontenc}

\usepackage[colorinlistoftodos, color=green!40, prependcaption]{todonotes}

\usepackage{amsthm}
\usepackage{mathtools}
\usepackage{physics}
\usepackage{xcolor}
\usepackage{graphicx}
\usepackage{tikz-cd}
\usepackage{bbold}

\usepackage{adjustbox}
\usepackage{placeins}

\usepackage{csquotes}
\usepackage{booktabs}

\usepackage[caption=false]{subfig}
 \makeatletter
\renewcommand{\fnum@figure}{FIG. \thefigure}
\makeatother
\usepackage{tikz}
\usetikzlibrary{decorations.markings}

\usepackage{amsmath,mleftright}
\usepackage{xparse}

\NewDocumentCommand{\evalat}{sO{\big}mm}{%
  \IfBooleanTF{#1}
   {\mleft. #3 \mright|_{#4}}
   {#3#2|_{#4}}%
}

\usepackage{amsthm}

\usepackage{extpfeil}
\usepackage{tikz-cd}
\usepackage{svg}
\usepackage{nicematrix}
\usepackage{comment}

\usepackage{titlesec}
\titleformat{\section}
  {\normalfont\bfseries}{\thesection}{1em}{}
\titleformat{\subsection}
  {\normalfont\bf}{\thesubsection}{1em}{}{\normalfont\bfseries}
\titleformat{\subsubsection}
  {\normalfont\bfseries}{\thesubsubsection}{1em}{}
\renewcommand{\thesubsection}{\thesection.\arabic{subsection}}
\renewcommand{\thesection}{\arabic{section}}
\renewcommand{\thesubsubsection}{\thesubsection.\arabic{subsubsection}} 

\usetikzlibrary{arrows}

\usepackage{leftindex}
\makeatletter
\newtheorem*{rep@theorem}{\rep@title}
\newcommand{\newreptheorem}[2]{%
\newenvironment{rep#1}[1]{%
 \def\rep@title{#2 \ref{##1}}%
 \begin{rep@theorem}}%
 {\end{rep@theorem}}}
\makeatother
\newreptheorem{proposition}{Proposition}

\DeclareMathOperator{\PAut}{{\bf PAut}}

\newcommand{\NN}{\ensuremath{\mathbb{N}}}

\newcommand{\C}{\ensuremath{\mathcal{C}}}

\newcommand{\F}{\ensuremath{\mathbb{F}}}

\newcommand{\cL}{\mathcal{L}}

\newcommand{\Enc}{\operatorname{E}}

\newcommand{\fs}{\mathfrak{s}}

\newcommand{\tF}{\widetilde{F}}
\newcommand{\tC}{\widetilde{C}}

\DeclareMathOperator{\Pic}{\ensuremath{\text{\normalfont{Pic}}}}

\DeclareMathOperator{\DDiv}{Div}

\DeclareMathOperator{\Hil}{\text{\normalfont{Hil}}}

\DeclareMathOperator{\Gal}{\text{\bf{Gal}}}

\DeclareMathOperator{\Con}{\text{\normalfont{Con}}}  

\begin{document}

\title{Quantum group codes for non-Clifford logic: \\enhanced decoding, addressability and parallelizability}

\author[1]{Jean Gasnier}
\author[1,2]{Virgile Gu\'emard}

\affil[1]{Institut de Math\'ematiques de Bordeaux, UMR 5251, Universit\'e de Bordeaux, France}
\affil[2]{Naquidis Center, Institut d’Optique Graduate School, 91127, Palaiseau, France}

\date{\today}

\maketitle

\begin{abstract}
We introduce a framework based on classical quasi group codes to define a class of quantum CSS codes, called \textit{quantum group codes}, supporting transversal multi-control-$Z$ gates which are both addressable and parallelizable, thus allowing to efficiently implement circuits composed of non-Clifford gates at the logical level. Building on this, we use a lifting procedure of classical AG codes established from class field theory to construct good quantum group codes with improved decoding complexity and logical multi-control-$Z$ gate parallelizability. More precisely, on input a good quantum AG code over the alphabet $\F_q$ with transversal $\mathsf{C}^m\mathsf Z$ gate, we apply this lifting procedure to its underlying classical AG code and obtain a quantum group code over the alphabet $\F_{q^2}$ supporting a transversal $\mathsf{C}^m\mathsf Z$ gate as well as addressable and parallelizable $\mathsf{C}^{m-1}\mathsf Z$ gates. In addition, this quantum code admits a quasi-quadratic time decoder with a linear decoding radius. This is to be compared with the previous quantum AG codes which have a cubic-time decoder. Hence, our work implies a decrease of the time complexity of state-of-the-art magic-state distillation protocols by an almost linear factor.
\end{abstract}
\tableofcontents

\section{Introduction}

Existing quantum fault-tolerant computation schemes \cite{Aharonov97,Gottesman2013,8555154} suffer from large space and time overhead, still rendering a practical quantum computer out of reach. The study of quantum error correcting codes, and in particular Calderbank-Shor-Steane codes \cite{Calderbank1996,Shor1995,Steane1996}, supporting transversal gates lies at the center of much recent progress. Broadly speaking, we say that a quantum code admits a transversal gate $U$ if there exists a short depth physical quantum circuit composed of local gates whose effect is to implement $U$ at the logical level. The low connectivity of the physical circuit prevents errors from propagating, making transversal gates naturally fault tolerant.  While no quantum code supports a universal set of transversal gates \cite{PhysRevLett.102.110502}, it is a major challenge to design codes admitting certain transversal gates up to the third level of the Clifford hierarchy such as the $T$ and the $CCZ$ gates\cite{PhysRevA.95.012329}.\par

Major efforts have been invested both in the study of fault-tolerant properties of quantum low-density parity-check (LDPC) codes and non-LDPC codes. While the class of LDPC codes is one of the leading candidates for physical implementations and extensive efforts have already been invested on designing LDPC codes with transversal Clifford and non-Clifford gates \cite{PhysRevLett.97.180501,Bombin_2015,PhysRevA.91.032330,PhysRevLett.111.090505,Kubica_2015,PhysRevA.100.012312}, with recent progress on high rate codes~\cite{zhu2024nonclifford,scruby2024quantumrainbowcodes,zhu2025topological,golowich2024quantumldpccodestransversal,lin2024transversalnoncliffordgatesquantum,li2026transversal}, the parameters and computation capabilities of these codes appear to be still limited. Alternatively, relaxing the LDPC constraint has proved fruitful to understand which properties of a code allows for efficient fault-tolerant computation. Currently, parameters and computational capacities of non-LDPC codes surpass those of LDPC codes, which motivates our focus on those codes in this work.

First, non-LDPC codes supporting non-Clifford transversal gates are a key component of magic state distillation (MSD) \cite{PhysRevA.71.022316}, a protocol that enables universal fault-tolerant quantum computation when combined with gate teleportation. MSD is a very resource-intensive technique and improving of its time and space overhead, the latter being defined as the ratio of input to output magic states of a protocol, is the subject of major research efforts. First, it was shown that generator matrices of punctured classical Reed-Muller and Reed-Solomon codes yield quantum codes over prime-dimensional qudits adapted for MSD \cite{PhysRevX.2.041021,PhysRevA.86.052329,Krishna2019}. These techniques have been generalized to algebraic geometry (AG) codes, producing asymptotically good binary quantum codes supporting non-Clifford transversal gates, thus allowing for asymptotically constant space-overhead MSD~\cite{ Wills2025,golowich2024asymptoticallygoodquantumcodes,nguyen2024goodbinaryquantumcodes}.

In all prior work mentioned above, the unique transversal gate supported by the LDPC or non-LDPC codes is implemented by acting on all the physical qudits at the same time, and its effect is to act on all or a large fraction of the logical qudits. For high rate quantum codes, the algorithmic practicality and application outside of MSD of such global gates appear to be limited. What is rather needed is a precise control over which logical qudits are operated on. Therefore, a new direction of study has been initiated \cite{he2025addressable} on quantum codes supporting addressable non-Clifford gates. Namely, a $m$-qudit gate is addressable if a large fraction of $m$-subsets of logical qudits can be acted upon transversally\footnote{Note that we are only interested in transversal gates, so that addressability is meant to imply transversality in this work} via a short depth physical circuit of local gates. In \cite{he2025addressable,he2025good}, the addressable orthogonality framework accounts for codes with these properties and proves the existence of such codes with asymptotically good parameters. Building on this work, it was shown in \cite{guemard2025good} that there exists a good code family that supports addressable gates in arbitrary levels of the Clifford hierarchy, and more importantly that these gates can be parallelized to a high degree. These results lead to a substantial reduction in the depth overhead of multi-control-$Z$ circuits. More precisely, with these codes, the minimal depth of any logical $\mathsf C^{m-1}\mathsf Z$-circuit, $m>1$, involving logical qudits from $m$-distinct code blocks is upper bounded by $O(k^{m-1})$, where $k$ is the code dimension\footnote{This scaling is optimal for dense circuits but can be readily improved from the results of \cite{guemard2025good} by considering specific circuits.}.

\paragraph{Main results.} The recent work on MSD and addressability presented above may still be improved at several levels. First, the MSD scheme of \cite{Wills2025} is based on quantum codes admitting a decoder with a cubic time-complexity in the length of the code. This implies that the MSD protocol \cite{Wills2025} inherits a cubic time-complexity, thus leaving room for progress. The constructions of asymptotically good addressable codes suffer from the same high decoding complexity. Second, the only existing family of quantum codes with linear distance and fully addressable gates is built from a very specific family of classical codes, namely transitive AG codes defined over the tower of function fields from~\cite{Stichtenoth2006}. It remains interesting to study different families of such codes.

In this work, we address both of these problems. 
First, we identify sufficient properties of a quantum code to admit parallelizable gates. Such quantum codes can be directly built from classical quasi group codes with multiplication properties. Then, we leverage recent lifting results on AG codes to convert good quantum AG codes with a unique global transversal gate, such as those used in MSD protocols, into quantum codes with addressable and parallelizable non-Clifford gates. The lifting is done via an abelian function field extension developed in \cite{CouGas26} using class field theory, which naturally endows the linear codes with a linear-time encoder and a quadratic-time decoder, up to logarithmic factors, with a linear decoding radius. In this work, we verify that this transformation preserves the properties necessary to build a quantum CSS code with transversal non-Clifford gates. Moreover, the inherited free-module structure of the code is exactly suited for addressability and parallelizability of the logical gates. Importantly, the quadratic time decoder can be directly used to decode the quantum codes we construct up to errors of linear size.\par

We can summarize our results in the following theorem.

\begin{theorem}\label{Theorem, main}
Let $t>4$ be a power of two, and set $q =t^4$. Then, for any $1<m\leqslant \frac{t}{2}-1$ there exists:

\begin{enumerate}[(1)]
    \item a sequence of $q$-dimensional qudit quantum group codes with asymptotic parameters 
\[ 
[[n,\Theta (n/\log n) ,\Theta(n)]]_{q},
\]
and admitting transversal $\mathsf C^{\tilde m-1} \mathsf Z$-gates that are fully addressable for all $1<\tilde m\leqslant m$. Moreover, the depth of any logical circuit of $\mathsf C^{\tilde m-1} \mathsf Z$-gates acting across the $m$ code blocks of $\mathcal{Q}^{\otimes \widetilde m}$ is upper bounded by $O(k^{\tilde m-1})$, where $k$ denotes the dimension of $\mathcal Q$.

\item a sequence of $q$-dimensional qudit quantum group codes with asymptotic parameters 
\[ 
[[n,\Theta (n) ,\Theta(n)]]_{q},
\]
and admitting a transversal $\mathsf C^{\tilde m }\mathsf Z$-gate as well as orbit-wise addressable transversal $\mathsf C^{\tilde m-1 }\mathsf Z$ gates across $\Theta(\log(n))$ orbits of size $\Theta(n/\log(n))$ for $1<\tilde m \leqslant m$.
\end{enumerate}

In both cases, the quantum code family admits a $O(n^2\cdot \operatorname{polylog}(n))$-time decoder\footnote{The polylogarithmic factor of the decoder is $\log^4(n)$} with a linear decoding radius.
\end{theorem}

Note that for simplicity of exposition, we focus on the characterization of inter-block multi-control-$Z$ circuits. Nevertheless, our result can be directly extended to circuits of  $C^{m-1}Z$ gates acting on a single codeblock of $\mathcal Q$, at the expense of a factor in the circuit depth linear in $m$.\par

As can be seen from the parameters of the codes in item \textit{(1)} and \textit{(2)}, our construction allows for a trade-off between the degree of addressability of the multi-control-$Z$ gates and the dimension of the quantum code. In item \textit{(1)}, we obtain the highest degree of addressability, while in item \textit{(2)}, there is still a strong degree of addressability, but we emphasize that there is still a transversal gate with global action. Namely, by acting on all the physical qubits across the $m$ blocks of $\mathcal Q$, we obtain a global action on all the logical qudits: $\mathsf C^{\tilde m-1} \mathsf Z^{\otimes n}=\overline{\mathsf C^{\tilde m-1} \mathsf Z}^{\otimes k}$.\par

Given a $[[n,k,d]]$ CSS code $C$ with a transversal $\mathsf{CCZ}$ gate, let $\gamma\coloneq \frac{\log(n/k)}{\log d}$ denotes the overhead exponent of the associated magic-state-distillation (MSD) protocol. An MSD protocol is said to have constant overhead when its associated asymptotic overhead exponent is $\gamma=0$. State-of-the-art magic state distillation protocols \cite{Wills2025} have a time complexity in $O(n^3)$. Here, we obtain the following corollary of Theorem \ref{Theorem, main}. 

\begin{corollary}\label{Corollary MSD}
Let $q=t^4$ with $t\geqslant 6$. Then, there exists a $[[n,\Theta(n),\Theta(n)]]_q$ quantum code family allowing for a constant overhead distillation protocol of qubit $\mathsf{CCZ}$-magic states with an $O(n^2\cdot \operatorname{polylog}(n))$ time complexity.
\end{corollary}

\paragraph{Organization of the article.} In Section \ref{Section Preliminaries}, we review some preliminary notions related to linear codes. In particular, quasi group codes, algebraic geometry codes and their multiplicative properties are discussed.\par
In Section \ref{Section quantum group codes}, we introduce the notion of quantum group codes and we present our generic construction of a quantum code having advantageous properties for multi-control-$Z$ circuits.\par
In Section \ref{Section: good quantum},  we present our main application. Based on a lifting procedure for AG codes developed in \cite{CouGas26}, we construct quantum group codes with good parameters and whose automorphism group contains an abelian subgroup of large order acting freely. This leads to the families of codes presented in Theorem \ref{Theorem, main}.\par

\paragraph{Acknowledgements.} V.G. would like to thank Anthony Leverrier for useful discussions. V.G also acknowledges the Plan France 2030 through the project NISQ2LSQ ANR-22-PETQ-0006, and funding from the Naquidis Innovation Center, under which part of this work was produced.

\section{Preliminaries: linear codes and algebraic geometry codes}\label{Section Preliminaries}

\subsection{Linear codes}

Let $\F_q$ be a finite field with $q$ elements. A linear $[n,k,d]_q$-code $C$ is a $k$-dimensional subspace $C \subseteq \F_q^n$ with minimum (Hamming) distance $d$. We are interested in codes having an additional structure of group module. Throughout, we write $[n]\coloneq \{1,\ldots n\}$.

\subsubsection{Group codes and quasi group codes}

Given a permutation $\sigma \in S_n$ and a vector $x = (x_i)_{i\in [n]}=(x_1,\dots,x_n)\in \F_q^n$, we denote by $x\cdot \sigma$ the vector 
\[x\cdot \sigma =(x_{\sigma(1)},\dots,x_{\sigma(n)})\]
This defines a right action of $S_n$ on $\F_q^n$. Recall that one says that the action of a subgroup $G\subseteq S_n$ on $[n]$ is
\begin{enumerate}
    \item \emph{transitive} if, for any given $i,j\in [n]$, there exists $\sigma\in G$ such that $\sigma(i) = j$.
    \item \emph{free} if, for every $i\in [n]$, the only $\sigma \in G$ such that $\sigma(i)=i$ is the identity element.
\end{enumerate}
Moreover, if the action of $G$ is free, it means its order divides $n$, i.e. $\ell = \frac{n}{|G|}$ is an integer.\par
Let $C\subseteq \F_q^n$ be a linear code and define
\[\PAut(C) = \{\sigma \in S_n \mid C\cdot \sigma = C\}\]
the \emph{permutation automorphism group} of $C$, i.e. the group of permutations of $S_n$ stabilizing $C$.
One says that a code $C$ is \emph{transitive} if $\PAut(C)$ acts transitively on $[n]$. By definition, any subgroup $G \subseteq \PAut(C)$ acts on $C$ (on the right), as well as on $[n]$ (on the left).\par
An interesting case arises when there exists a subgroup $G\subseteq \PAut(C)$ whose action on $[n]$ is both transitive and free. In this case, as shown in~\cite{BerRioSim09}, the code $C$ is a \emph{$G$ code}, i.e. is isomorphic to a right ideal of the group algebra $\F_q[G]$. For instance, if $G$ is cyclic, this amounts to say that $C$ is a cyclic code. Finally, we will be interested in the following slight generalization of $G$ codes.

\begin{definition}[Quasi-$G$ code]\label{def-quasi-G-codes}
Let $C$ be a $q$-ary code of length $n$ and let $G$ be a group. Then, $C$ is a \textit{quasi-$G$ code} if there exists a subgroup $H\subseteq \PAut(C)$ isomorphic to $G$ whose action on $[n]$ is free. If $G$ is abelian, we will also say that $C$ is a \textit{quasi-abelian code}.
\end{definition}

According to~\cite{BorWil23}, this amounts to say that $C$ is isomorphic to a right $G$-module inside $\F_q[G]^\ell$ (on which $G$ acts by multiplication on the right on each coordinate), where $\ell = \frac{n}{|G|}$ is the number of orbits of $G$ in $[n]$. 

An important property of quasi group codes is the following:
\begin{lemma}
Let $G$ be a group. Then the dual of a quasi-$G$ code is a quasi-$G$ code.
\end{lemma}
\begin{proof}
Let $C$ be a quasi-$G$ code and let $y\in C^\perp$. For any $x\in C$ and $\sigma \in G$, one has
\[\langle x,y\cdot \sigma \rangle = \langle x\cdot \sigma^{-1},y\rangle =0.\]
\end{proof}
\begin{remark}
One can define a left action of $S_n$ on $ \F_q^n$ by for $x\in \F_q^n$ and $ \sigma \in G, \sigma \cdot x \coloneq x\cdot \sigma^{-1}. $ Since $G\subseteq S_n$ is stable under inversion, any quasi-$G$ code $C$ is stable under the induced action of $G$ on the left. In this setting, $C$ is also isomorphic to a left submodule of $\F_q[G]^\ell$. This is the image of the right $G$-module of $\F_q[G]^\ell$ isomorphic to $C$ under the antimorphism defined by $\sigma  \mapsto \sigma^{-1}$. This shows that there is no preferred side for working with quasi-$G$ codes, but that one needs to pay attention to the choice of convention. This is particularly true when $G$ is non-commutative.
\end{remark}

\paragraph{Quasi-$G$ codes and $G$-linear algebra.} One interesting property of quasi-$G$ codes is the possibility to define them via a generator matrix with coefficients in $\F_q[G]$. Indeed, let $C$ be a quasi-$G$ code with parameters $[\tilde{n},\tilde{k},d]$, let $n = \frac{\tilde{n}}{|G|}$ and let $M$ be the corresponding right $G$-module of $\F_q[G]^n$. Let $g_1,\ldots,g_k \in \F_q[G]^n$ be a set of generators of $M$ of minimum size. Then 
\[M = \left\{\: \sum_{i=1}^k g_i \cdot \alpha_i \: :\: \alpha \in \F_q[G]^k \: \right\}\]
or in other words $M$ is the image of $\mathbf M \in \operatorname{Mat}_{n,k}(\F_q[G])$ (under multiplication on the right) where the columns of $\mathbf M$ are the generators $(g_i)_{i\in [k]}$.

The matrix $\mathbf M$ describes a $G$-linear morphism, which is in particular a linear morphism. By fixing an order on the elements of $G$, one can associate with $\mathbf M$ a matrix $\widetilde{\mathbf M}$ of size $n\cdot|G| \times k\cdot |G|$ with coefficients in $\F_q$ with a particular shape. Explicitly,  one computes the matrix  $\widetilde{\mathbf M}\in \operatorname{Mat}_{\tilde{n}, k\cdot|G|}(\F_q)$ by replacing every entry $M_{ij}\in \F_q[G]$ by its left-regular representation, or equivalently the matrix of the multiplication by $M_{ij}$ on the left. 
One would want the columns of $\widetilde{\mathbf{M}}$ to generate the code $C$. However, the orbits of the action of $G$ on $[n]$ need not be composed of adjacent integers, so this will not be the case in general. Still, there always exists a permutation $\sigma \in S_n$ such that the columns of $\widetilde{\mathbf{M}}$ generate $C\cdot \sigma$, a code permutation equivalent to $C$.

Note also that the columns of $\widetilde{\mathbf M}$ need not be linearly independent in general. In fact, $\widetilde{\mathbf M}$ has rank $k\cdot|G|$ if and only if $M$ is \emph{free}, i.e. is isomorphic to $\F_q[G]^k$ as a right $\F_q[G]$-module, in which case $\tilde k = k\cdot|G|$.

In a slightly abusive manner, we will call $\mathbf M$ a generator matrix of $C$ as well. To avoid confusion with traditional generator matrix, we will emphasize on the ring of definition of the coefficients of the matrix.

\begin{remark}\label{rem-left-G-matrix}
An equivalent construction can be defined for the left $G$-module $M'$ associated with $C$. 
In that case, $M'$ is the image (under multiplication on the left) of a matrix $\mathbf M'\in \operatorname{Mat}_{k,n}(\F_q[G])$. One can then associate with $\mathbf M'$ a matrix $\widetilde{\mathbf{M}}'$ with coefficients in $\F_q$ by replacing each entry by its right-regular representation.
One needs to carefully choose the basis of $\F_q[G]$ in order to compute $\widetilde{\mathbf M}'$ in a coherent manner. The detailed process is available in~\cite[Section~4.4]{GasnierPhD}.
\end{remark}

In order to stick to the writing conventions in coding theory, we will choose to see quasi group codes as modules on the left as much as possible in what follows.

\begin{remark}
When $M$ is free, i.e. isomorphic to $\F_q[G]^k$, the representation of $C$ by $\mathbf M \in \operatorname{Mat}_{n,k}(\F_q[G])$ requires $n\cdot k\cdot |G|$ coefficients in $\F_q$ instead of $\tilde{n}\cdot\tilde{k} = n\cdot k\cdot |G|^2$ coefficients for a classical generator matrix of $C$. One can see from this favorable situation that quasi-$G$ codes possess shorter representations than generic codes.
\end{remark}

\paragraph{Encoding quasi-$G$ codes.} Let $C$ be a quasi-$G$ code with parameters $[\tilde{n},\tilde{k},d]$, let $n = \frac{\tilde{n}}{|G|}$ and let $M$ be the corresponding left $G$-module of $\F_q[G]^n$. Let $g_1,\ldots,g_k \in \F_q[G]^n$ be a set of generators of $M$ of minimum size, and $\mathbf{M}\in \operatorname{Mat}_{k,n}{\F_q[G]}$ the corresponding generating matrix of $C$. 

One may sometimes take advantage from the $G$-linearity of $M$ to encode $C$ faster than by performing a naive evaluation of a generator matrix, at the cost of $O(\tilde{k}\cdot \tilde{n})$ operations. Essentially, a product in $\F_q[G]$ can be computed at the cost of at most $\C_q(|G|\log|G|)$ operations using the Fast Fourier Transform~\cite[Theorem~2]{CouGas26}, where $\C_q$ is a constant depending only on $q$. As such, the encoding of $C$ in this fashion requires at most $\C_q(n\cdot k\cdot|G|\log|G|)=\C_q(k\log|G|\tilde{n})$ operations.

Thus, if $k\log |G|$ becomes negligible relative to $\tilde{k}$, then this technique allows for a faster encoding of $C$. This is particularly true when $C$ is a free $\F_q[G]$-module and $\tilde{k} = k\cdot |G|$.

\begin{proposition}\label{prop-encoding-quasi-G-codes}
Let $C$ be a quasi-$G$ code with parameters $[\tilde{n},\tilde{k},d]$, let $n = \frac{\tilde{n}}{|G|}$ and let $M$ be the right $G$-module of $\F_q[G]^n$ associated with $C$. Assume that $M$ is free, i.e. that $M$ is isomorphic to $\F_q[G]^k$ as an $\F_q[G]$-module, where $k = \frac{\tilde{k}}{|G|}$. Then $C$ can be encoded at the cost of at most $\C_q(n\cdot k\cdot|G|\log|G|)$ operations, where the positive constant $\C_q$ depends only on $q$.

Assume further that there exists a positive constant $\alpha_q$ such that
\[ n \leqslant \alpha_q \log|G|.\]
Then, $C$ can be encoded in quasi-linear time in its length $\tilde{n} = n\cdot |G|$.
\end{proposition}
\begin{proof}
We prove the second assertion. Recall that $k\leqslant n \leqslant \alpha_q\log|G|$.
Thus $C$ can be encoded at the cost of at most $\alpha_q^2\cdot \C_q ( |G|\log^3|G|)$ operations. Since the length of $C$ is $\tilde{n} = n\cdot |G|\geqslant |G| $, this complexity is quasi-linear in $\tilde{n}$.
\end{proof}

\subsubsection{Codes with multiplication properties}

Central for this work are codes satisfying different levels of self-orthogonality. Given two vectors $x$ and $y$ of $\F_q^n$, their Schur product is defined as the vector $x  \star   y \coloneq (x_1 y_1,\dots,x_ny_n)$. Given a linear code $C\subseteq \F_q^n$, we also define its $n$-fold Schur product by $$
C^{ \star m} =\operatorname{Span} \{\, x^1  \star  x^2  \star  \dots  \star  x^m : x^1, x^2, \dots, x^m \in C \,\}.$$
Moreover, given a vector $u  \in \F_q^n$ with only nonzero components, we define $
u  \star  C \coloneq \{\, u\star x: x \in C \,\}$. The codes $C$ and $u  \star  C$ have the same length, dimension, and minimum distance. The permutation automorphism groups $\PAut(C)$ and $\PAut(u  \star  C)$ are, in general, non-isomorphic.

\begin{definition}[Multiplication property]
Let $m \geqslant 1$ be an integer. We say that a linear code $C$ satisfies the \emph{$m$-multiplication property} if there exists a fixed vector $u$ with only nonzero components such that $u  \star  C^{ \star m} \subseteq C^\perp$.
\end{definition}

For any vector space $\F_q^n$, we call the vector with all entries equal to one the \emph{all-one vector}.

\begin{lemma}\label{lemma:multiplication-property}
Let $C$ be a code satisfying the $m$-multiplication property and containing the all-one vector. Then, $C$ satisfies the $\widetilde m$-multiplication property for every $1 \leqslant \widetilde m \leqslant m$.
\end{lemma}
The proof can be found in \cite{guemard2025good}. In that case, notice that $u  \star  C \subseteq C^\perp$, and since the all-one vector belongs to $C$, it follows that $u \in C^\perp$.

\subsubsection{Puncturing and shortening}

We shall later construct quantum codes from classical codes. This will rely on puncturing and shortening classical codes and we hence recall the definition of these code transformations. Throughout this section, we consider a linear code as a subset of $\F_q^N$, where $N$ is an index set of cardinality $n$. We may for example think of $N$ as being equal to $[n]$.

\begin{definition}[Punctured and shortened codes]
Let $L\subseteq N$ and denote the complement of $L$ in $N$ as $\overline L\coloneq N\setminus L$. Given a vector $c\in\F_q^N$, its restriction to the indices of $L$ is denoted $c_L\coloneq (c_i)_{i\in L}\in \F_q^L$. Then, for a classical code $C\subseteq\F_q^N$, we define
\begin{enumerate}[(1)]
\item The \textit{$L$-punctured code} $C_{\overline{L}}:=\{\: c_{\overline{L}}: c\in C\: \}\subseteq\F_q^{\overline L}$.
\item The \textit{$L$-shortened code} $C_{L=0}:=\{\,c_{\overline{L}}: c\in C,\; c_L=0\,\}\subseteq\F_q^{\, \overline L}$.
\end{enumerate}
\end{definition}
We will also use the following duality relations between puncturing and shortening as well as useful parameter estimates.

\begin{lemma}\label{Lemma puncturing}
Let $C\subseteq \F_q^N$ be an $[n,k,d]$ code and let $L\subseteq N$ be of cardinality $\ell$. Then, $(C^\perp)_{L=0}=(C_{\overline{L}})^\perp$ and $(C^\perp)_{\overline{L}}=(C_{L=0})^\perp$. Moreover we have the following parameter estimates.
\begin{enumerate}[(1)]
\item If $\ell<d$, then $C_{\overline{L}}$ has dimension $k$, and $(C^\perp)_{L=0}$ has dimension $n-\ell-k$.
\item $C_{L=0}$ has distance $d$, and $C_{\overline{L}}$ has distance at least $d-\ell$.
\end{enumerate}
\end{lemma}
See Theorem 1.5.7 of \cite{Huffman_Pless_2003} for a proof.

\subsection{Algebraic geometry codes}

\subsubsection{Background on function fields}

\paragraph{Function fields.}
Let $\F_q(z)$ be the rational function field. A function field $F / \F_q$ over $\F_q$ is a finite extension of $\F_q(z)$. It is well-known that there is an equivalence of categories between the category of function fields over $\F_q$ and the category of smooth projective curves over $\F_q$~\cite[Section 7, Rem. 3.14]{Liu02}. Through this equivalence, the notion corresponding to a (closed) point is a \emph{place} $P$ of $F$, i.e. the maximal ideal of a valuation ring of $F$. We denote by $\nu_P$ the corresponding valuation and $\mathcal{O}_P$ its valuation ring.

Let us assume that the field of constants of $F$ (i.e. the algebraic closure of $\F_q$ in $F$) is $\F_q$. The \emph{residue field} at $P$ is the field $K_P = \mathcal{O}_P / P$. The \emph{evaluation} of an element $f \in \mathcal{O}_P$ at $P$ is the image of $f$ under the canonical map $\mathcal{O}_P \to \mathcal{O}_P / P$. It is denoted by $f(P)$. The \emph{degree} of a place is the degree of the extension $K_P/\F_q$. A place is called \emph{rational} if its residue field is isomorphic to $\F_q$, or equivalently, if it has degree one.

\paragraph{Extensions of algebraic function fields.}
Let $F$ be a function field and assume that $\tF$ is an extension of $F$. 
Let $Q$ be a place of $\tF$, then $P:=Q\cap F$ is a place of $F$. One says that $Q$ \emph{lies above} $P$ and we denote this relation by $Q \mid P$.
We say that a place $P$ of $F$ \emph{splits completely} in $\tF$ if the number of places $Q$ of $\tF$ such that $Q \mid P$ is exactly $[\tF : F]$.

Assume that the extension $\tF/F$ is Galois, with Galois group $G$. The Galois group acts transitively on the set of places lying above a given place of $F$. 

Let $P$ be a place of $F$ and $Q$ a place of $\tF$ lying above $P$. One can see that the residue field $K_Q$ is an extension of $K_P$. One defines the \emph{inertia degree} of $Q$ as \[f_Q:=[K_Q:K_P].\]
Since $\tF/F$ is Galois, it does not depend on the choice of $Q$ above $P$, and will be denoted by $f_P$ as well.\par

Define the \emph{decomposition group} of $Q$ as
\[D(Q/P) \coloneq \{\sigma \in G \mid \sigma(Q) = Q\}.\]
There exists a canonical group morphism
\[D(Q/P) \longrightarrow \Gal(K_Q/K_P)\]
whose kernel is $I(Q/P)$ the inertia group of $Q$. One defines the \emph{ramification index} of $Q$ as $e_Q = |I(Q/P)|$. It does not depend on the choice of $Q$ as well because $\tF/F$ is Galois, and therefore can be denoted by $e_P$. Let $d_P$ be the number of places of $\widetilde F$ lying above $P$, then one has
\[d_P\cdot e_P\cdot f_P=[\tF:F].\]
Thus, $P$ splits completely in $\tF$ if and only if $e_P=f_P=1$.

\paragraph{Riemann-Roch spaces.}
A \emph{divisor} $D$ in a function field $F / \F_q$ is a formal sum of finitely many places with integer coefficients, namely $D = \sum n_P P$. A divisor $D$ is said to be \emph{positive} or \emph{effective}, written $D \geqslant 0$, if all of its coefficients are positive integers. Given an element $f \in F$, we write $(f)$ for its associated \emph{principal divisor}. Moreover, we denote the \emph{Riemann-Roch space} associated with a divisor $D$ by
\[
\mathcal{L}(D) = \{ f \in F : (f) + D \geqslant 0 \} \cup \{ 0 \}.
\]
Writing $D = D_{\text{pos}} - D_{\text{neg}}$ as the difference of two effective divisors, we see that $f \in \mathcal{L}(D)$ if and only if $f$ has poles of limited order at the places of $D_{\text{pos}}$ and zeros of sufficient order at the places of $D_{\text{neg}}$. If we have another divisor $D' \geqslant D$, then $\mathcal{L}(D')$ allows functions with fewer zeros and more poles, and thus $\mathcal{L}(D) \subseteq \mathcal{L}(D')$.

\paragraph{Picard group.}

Let $D$ and $D'$ be two divisors of $F$, one defines them to be \emph{equivalent} if $D-D'$ is a principal divisor, which is denoted by $D\sim D'$. One then defines the \emph{Picard group} of $F$
\[\Pic(F) = \DDiv(F)/\sim\]
the group of equivalence classes of divisors of $F$. Since every principal divisor has degree $0$, one can extend the notion of degree to equivalence classes. Then, one defines
\[\Pic^0(F) = \{c \in \Pic(F) \mid \deg c = 0\}\]
the subgroup of classes of degree $0$.

\paragraph{Lifting divisors and differentials.}
Let $F/\F_q$ be a function field, and let $\widetilde F/F$ be a separable extension of $F$. Let $D$ be a divisor of $F$. One can define the \emph{conorm} in $\widetilde{F}$ of a divisor of $F$ by
\[
\Con_{\tF/F}\left(\sum n_P P \right) = \sum_P \sum_{Q\mid P} e_Q \cdot n_P Q
\]
where $e_Q$ denotes the ramification index of $Q$. It follows from this definition that, for two divisors $D$, $D'$ of $F$, if $D\leqslant D'$ then $\Con_{\tF/F}(D)\leqslant \Con_{\tF/F}(D')$.

As shown in~\cite[Corollary~3.1.14]{Stichtenothbook}, the degree of the conorm satisfies the following relation
\[
\deg \Con_{\tF/F}(D) = \frac{[\tF:F]}{[K:\F_q]}\deg D.
\]
where $K$ is the field of constants of $\tF$.
We denote by $ \Omega_F$ the space of differentials of $F$. Given a differential $\omega \in \Omega_F$, there exists a unique differential of $\tF$ extending $\omega$ in the sense of~\cite[Theorem~3.4.6]{Stichtenothbook}, that we call the \emph{cotrace} of $\omega$ in $\tF/F$, and denoted by $\operatorname{Cotr}_{\tF/F}(\omega)$.

The divisor of the cotrace of $\omega$ is closely related to the conorm of the divisor of $\omega$:
\[\left(\operatorname{Cotr}_{\widetilde{F}/F}(\omega)\right) = \Con_{\tF/F}((\omega))+\Delta_{\tF/F}\]
where $\Delta_{\tF/F}$ is the \emph{different} of the extension $\tF/F$ (see~\cite[Definition~3.4.3]{Stichtenothbook}). The divisor $\Delta_{\tF/F}$ is supported by places ramified in $\tF$ (see~\cite[Theorem~3.5.1]{Stichtenothbook}). In particular, it is zero if the extension is unramified.

\subsubsection{Algebraic geometry codes and their parameters}

Let $F/\F_q$ be a function field whose field of constants is $\F_q$ and let $P_1, \ldots, P_n$ be $n$ distinct rational places of $F$. Let $P = P_1 + \ldots + P_n$ be a divisor of $F$. Let $D$ be a divisor with ${\rm supp}\, P \cap {\rm supp}\, D = \emptyset$. We may consider the \textit{algebraic geometry (AG) code}
\[
C_{\mathcal{L}}(P, D) \coloneq \{ (f(P_1), \ldots, f(P_n)) \in \F_q^{n} \mid f \in \mathcal{L}(D) \}.
\]
Algebraic geometry codes admit the following lower bounds on their parameters due to the Riemann-Roch Theorem.

\begin{proposition}\label{proposition standard estimate AG codes}
   Let $g$ denotes the genus of the function field $F$, defined over $\F_q$, and let $P$ and $D$ be defined as above. Suppose that
    \[
        2g - 2 < \deg D < n.
    \]
Then, the code $C_\mathcal{L}(P,D)$ is an $[n,k,d]$-code such that $k= \deg D + 1 - g$ and $d \geqslant n - \deg D$. Moreover, the dual code of $C_\mathcal{L}(P,D)$ has minimum distance $d^\perp \geqslant \deg D - 2g +2$.
\end{proposition}
It is common to define \emph{the designed distance} of $C_\mathcal{L}(P,D)$ as \[d^*:=n-\deg D.\]
Now assume that $F$ is a Galois extension of a function field $F'/\F_q$, with Galois group $G$, in which $P_1, \ldots, P_n$ are unramified. Also assume that $G$ fixes the two divisors $P$ and $D$. Then, $G$ acts freely on the set of rational places $\{P_1,\ldots, P_n\}$ (because they are unramified), and $G$ stabilizes the Riemann-Roch space $\cL(D)$. This yields the following result.

\begin{proposition}\label{proposition AG code automorphism}
With the notation of the beginning of the paragraph, assume that $F$ is a Galois extension of a function field $F'/\F_q$, with Galois group $G$, in which $P_1, \ldots, P_n$ are unramified. Also assume that $G$ fixes the two divisors $P$ and $D$. Then $C_\cL(P,D)$ is a quasi-$G$ code.
\end{proposition}

\begin{proof}
For $\sigma \in G$, and $f\in \cL(D) $, we define
\begin{equation}\label{eq-action-G-code}
\sigma^{-1} \cdot\big(f(P_1), \ldots, f(P_n)\big)\coloneq \left((\sigma^{-1}\cdot f)(P_1), \ldots, (\sigma^{-1}\cdot f)(P_n)\right) = \big(f(\sigma  (P_1)), \ldots, f(\sigma( P_n))\big)\in C_\cL(P,H).
\end{equation}
By identifying the sets $[n]$ and $\{P_1,\dots, P_n\}$, we can then see $G$ as a subgroup of $S_n$. Moreover, $G$ stabilizes $C_\cL(P,D)$, thus $G$ is a subgroup of $\PAut(C_\cL(P,D))$. Finally, since the action of $G$ on the rational places is free, by Definition~\ref{def-quasi-G-codes}, the code $C_\cL(P,D)$ is a quasi-$G$ code.
\end{proof}

We end this section by giving criteria for $C_\cL(P,D)$ to possess the all-one vector and to satisfy the $m$-multiplication property.

\begin{proposition}
With the notation of the beginning of the paragraph, assume $D$ is effective. Then $C_\cL(P,D)$ contains the all-one vector. Moreover, let $m$ be an integer, and let $\eta \in \Omega_F$ be a differential of $F$ having a simple pole at every rational place of $P$. If
\[ (m+1)D \leqslant P+(\eta), \]
then the code $C_\cL(P,D)$ satisfies the $m$-multiplication property.
\end{proposition}

\begin{proof}
If $D$ is effective, $\F_q \subseteq \cL(D)$ and in particular $1 \in \cL(D)$. Thus the all-one vector belongs to $C_\cL(P,D)$. Now, let $u = (\operatorname{Res}_{P_i}(\eta))_{1\leqslant i \leqslant n}$ be the vector of residues of $\eta$ at $P$. From~\cite[Proposition 2.2.10]{Stichtenothbook}, since $\eta$ has simple poles at $P$, $u$ does not have a null component, and one has
\[C_\cL(P,D)^\perp = u \star C_\cL(P,P-D+\eta).\]
Since $D$ is effective, $\cL(D)^{\star m} \subseteq \cL(m\cdot D)$, and the hypothesis yields
\[u\star C_\cL(P,D)^{\star m} \subseteq C_\cL(P,D)^\perp.\]
\end{proof}

\section{Quantum group codes: parallelizable transversal gates}\label{Section quantum group codes}

\subsection{Quantum group codes and multi-control gates}

\paragraph{CSS codes.} An efficient method to protect quantum information against the effect of noise is provided by Calderbank-Shor-Steane (CSS) codes \cite{Calderbank1996,Steane1996,Stean1996Multiple}, a type of stabilizer quantum error correcting codes defined by two linear codes. Let us recall that the state space of an $n$ $q$-dimensional qudit system is $(\mathbb C^q)^{\otimes n}$, with computational basis states denoted $\{\ket{x}:x\in \F_q^n\}$. For our purpose $q$ denotes a power of two, although the discussion can be generalized to other field characteristics.

\begin{definition}[CSS code]\label{Def_CSS stabilizer code}
Let $C_0,C_1\subseteq \F_q^n$ be two linear codes such that  $C_0^\perp\subseteq C_1$. A quantum CSS code $\operatorname{CSS}(C_0,C_1)$ is the subspace of $(\mathbb{C}^q)^{\otimes n}$ given by\[\operatorname{Span}\left\{ \ket{c+C_0^\perp}\::\: c\in C_1 \right \}.\]
where we used the notation $\ket{c+C_0^\perp}\coloneq \frac{1}{|C_0^\perp|^{1/2}} \sum_{b\in C_0^\perp }\ket{c+b}$. Its parameters, noted $[[n,k,d]]_q$, are the \textit{code length} $n$, the \textit{dimension} $k = \dim(C_1/ C_0^\perp)$ and the \textit{distance} $d = \min(d_0, d_1)$, with
\begin{align*}
d_0=\min\limits_{c\in C_1\setminus C_0^\perp}|c|,\qquad
d_1=\min\limits_{c\in C_0\setminus C_1^\perp}|c|.
\end{align*}
An infinite family of CSS codes of growing size is said to be \textit{good} if its parameters scale asymptotically as $[[ n, \Theta(n), \Theta(n)]]$.
\end{definition}

Given a quantum CSS code $\mathcal Q\coloneq\operatorname{CSS}(C_0,C_1)$, we denote the quotient map $\pi:C_1\to C_1/C_0^\perp$. We may choose a set of representatives for the space $C_1/C_0^\perp$ and this choice is described via a \textit{section} of $\pi$, namely a map
\begin{equation}
\fs:C_1/C_0^\perp\to C_1
\end{equation}
such that $\pi\circ \fs=\text{id}$. Then, given an element $y\in C_1/C_0^\perp$ we define the following quantum state, up to a global normalization factor,
\begin{equation}
    \ket{y} \coloneq\sum_{c\in C_1 \::\: \pi(c)=y}\ket{c}.
\end{equation}
Note that given a representative $c$ of $y$, i.e. an element $c\in C_1$ with $\pi(c)=y$, the state $\ket{y}$ agrees with the common notation $\ket{c+C_0^\perp}$ used in Definition \ref{Def_CSS stabilizer code}. Using the section $\fs$, we can also write $\ket{y}=\ket{\fs(y)+C_0^\perp}$.\par
Finally, we define an \textit{encoding map} for $\mathcal Q$ as an isomorphism
$$\Enc :\F_q^k\xrightarrow{\sim} C_1/C_0^\perp.$$
That is, $\Enc $ fixes a choice of preferred basis for its codomain. Given such an encoding and any computational basis state $\ket{x}$, where $x\in \F_q^k$ , we then define its encoding $\overline{\ket x}$ with respect to $\Enc$ as the unit vector 
\begin{align}
    \overline{\ket{x}}&=\ket{ \Enc (x)}\\
    &=\ket{\fs\circ \Enc(x)+C_0^\perp}
\end{align}
We then refer to the image of the $i$th tensor component of $(\mathbb C^q)^{\otimes k}$ via $\Enc$ as the $i$th logical qudit.

\paragraph{Quantum $G$ codes.} Of importance for us will be quantum codes inheriting group properties from those of quasi-$G$ code. The following definition will be of particular relevance in the next section to estimate the existence of transversal multi-control-$Z$ gates.

\begin{definition}[Quantum $G$ code]\label{definition quantum G-code}
Let $G$ be a group and let $\mathcal Q=\operatorname{CSS}(C_0,C_1)$ be a quantum code such that $C_0,C_1\subseteq \F_q^n$ are quasi-$G$ codes. We see $C_0^\perp$ and $C_1$ as left $\F_q[G]$-modules. Then, we refer to $\mathcal Q$ as a \textit{quantum $G$ code} if there exists a section $\fs$ of the canonical map $\pi:C_1\to C_1/C_0^\perp$, i.e a $G$-module homomorphism
\begin{equation}
    \fs:C_1/C_0^\perp\to C_1
\end{equation}
such that $\pi\circ \fs=\operatorname{id}$. Moreover, we say that $\mathcal Q$ is a \textit{regular quantum $G$ code} if $C_1/C_0^\perp$ is isomorphic to $\F_q [G]$. Clearly, in that case, the dimension of $\mathcal Q$ is $k_{\mathcal Q}=|G|$. Finally, we say that $\mathcal Q$ is a \emph{free quantum $G$ code} if there exists an integer $r$ such that $C_1/C_0^\perp \simeq \F_q[G]^r$, in which case its dimension is $k_{\mathcal{Q}} = r|G|$. 
\end{definition}

\begin{remark}
Note that a section always exists in the sense of $\F_q$-vector spaces, but the stronger condition in Definition \ref{definition quantum G-code} requires $\fs$ to be compatible with the module structure, i.e $\fs\circ \sigma=\sigma\circ \fs$, for $\sigma\in G$. Also note that a surjective map to a free module always admits a section.
\end{remark}

\paragraph{Multi-control-$Z$ gates.} In this work, we consider quantum CSS codes whose main purpose is either to be used in magic state distillation protocols, or for circuit compilation. Central components of these protocols are non-Clifford gates, instances of which are multi-control-$Z$ gates. These particular gates constitute our main focus.\par
Let $m$ be a positive integer and $q$ be a power of $2$. For any $\gamma \in \F_q$, the action of a \textit{physical $(m-1)$-control-$Z$ gate} $\mathsf{C}^{m-1}\mathsf{Z}^\gamma$ on $m$ physical qudit registers in the computational basis is given as follows,
\[
\mathsf{C}^{m-1}\mathsf{Z}^\gamma\ket{x_0}\ket{x_1}\dots\ket{x_{m-1}}
\coloneq
\exp\left(i\pi \, \mathrm{tr}(\gamma \: x_0 \: x_1\dots x_{m-1})\right)
\ket{x_0}\ket{x_1}\dots\ket{x_{m-1}},
\]
$x_i\in \F_q$ for $i=0,\dots,m-1$, and $\mathrm{tr}\coloneq\mathrm{tr}_{\F_q/\F_2}$ denotes the trace map from $\F_q$ to the prime field $\F_2$.\par

Given a quantum CSS code $\mathcal{Q}$ of dimension $k$, we may encode information in the tensor product code $\mathcal{Q}^{\otimes m}$, which is composed of $m$ copies of $\mathcal{Q}$. We refer to each copy of the original code as a \emph{code block} of $\mathcal{Q}^{\otimes m}$. We label the physical qudit registers in each code block by the set $N$ of cardinality $n$.\par
A \textit{physical} $\mathsf{C}^{m-1}\mathsf{Z}$ gate on physical qudit registers corresponding to elements $P_0,P_1,\dots,P_{m-1}\in N$, respectively in code blocks $0,1, 2,\dots,m-1$, is denoted $\mathsf{C}^{m-1}\mathsf{Z}^\gamma[P_0,P_1,\dots,P_{m-1}]$, where $\gamma \in \F_q$. Likewise, given a set $L$ of cardinality $k$ and an encoding map $\Enc:\F_q^L\to C_1/C_0^\perp$ for $\mathcal Q$, a \textit{logical (interblock)} $\mathsf{C}^{m-1}\mathsf{Z}$ operator targeting logical qudits indexed by $Q_0,Q_1,\dots,Q_{m-1}\in L$ belonging respectively to code blocks $0,1, 2,\dots,m-1$, is denoted by $\overline{\mathsf{C}^{m-1}\mathsf{Z}^\gamma[Q_0, Q_1,\dots, Q_{m-1}]}$. Its effect is given on a computational basis state by
\[\overline{\mathsf{C}^{m-1}\mathsf{Z}^\gamma[Q_0, Q_1,\dots, Q_{m-1}]}\:\overline{\ket{x^0}}\:\overline{\ket{x^1}}\dots\overline{\ket{x^{m-1}}}\coloneq \exp\left(i\pi \, \mathrm{tr}(\gamma \: x^0_{Q_0} \: x^1_{Q_1}\dots x^{m-1}_{Q_{m-1}})\right)\overline{\ket{x^0}}\:\overline{\ket{x^1}}\dots\overline{\ket{x^{m-1}}},\]
where $x^i=(x^i_Q)_{Q\in L}$, $i=0,\dots,m-1$. For the quantum codes relevant for us, such a logical operator may be performed via a (non-unique) circuit of physical $(m-1)$-control-$Z$ gates. The minimal complexity of this circuit may vary from code to code but the ideal case is precisely when there exists such a circuit of depth one. In that case, a logical gate is said to admit a \textit{transversal implementation}, or simply to be \textit{transversal}.\par

\subsection{Sufficient conditions for parallelizable addressability}\label{section: Sufficient conditions for parallelizable addressability}

In this section, we introduce a generic construction of quantum codes that support transversal logical multi-control-$Z$ gates acting on any specified tuple of logical qudits of fixed size. We also analyze the extent to which these gates can be implemented in parallel.

\begin{theorem}\label{Theorem quasi-G-quantum}
Let $N$ be an index set of cardinality $n$ and let $G$ be a group\footnote{The index set $N$ can be thought of being equal to $[n]$ with $G\subseteq S_n$. } of order $\ell$ acting on $N$. Let $m\geqslant 2$ and $r$ be integers and let $C\subseteq \F_q^{N}$ be a quasi-$G$ code with parameters $[n,k,d]$ with $k>r\ell$, satisfying the $m$-multiplication property and the 2-multiplication property with respect to a vector $u=(u_i)_{i\in N}$. Suppose moreover that $C^\perp$ has distance $d^\perp > \ell$.
Let $L_1,\dots,L_r\subseteq N$ be $r$ distinct orbits of the action of $G$ on $N$, let $L=\cup_{i=1}^r L_i$ and let $\overline{L}\coloneq N\setminus L$. Let us define the two classical codes $C_0 \coloneq (C_{L= 0})^\perp \subseteq \F_q^{\overline{L}}$ and  $C_1 \coloneq C_{\overline{L}}\subseteq \F_q^{\overline{L}}$. Then, we have the following:
\begin{enumerate}[(1)]
    \item $C_0^\perp \subseteq C_1$ and the quantum CSS code $\mathcal Q\coloneq\mathrm{CSS}(C_0, C_1)$ is a free quantum $G$ code.
    \item $\mathcal Q$ has parameters $[[n_{\mathcal Q},k_{\mathcal Q},d_{\mathcal Q}]]$ with
\[
\begin{aligned}
n_{\mathcal Q} &= n -r\ell,\qquad k_{\mathcal Q} = r\ell, \qquad 
d_{\mathcal Q}  \geqslant \min\{d(C^{\perp}), d(C)\} -r \ell.
\end{aligned}
\]
\item There exists an encoding map $\Enc:\F_q^L\to C_1/C_0^\perp$ defining a basis of logical qudits $\{\overline{\ket{x}}=\ket{ \Enc (x)}: x\in \F_q^L\}$ such that given any subset $S \subseteq L$ and any $( m-1)$ elements $\sigma_1,\ldots,\sigma_{m-1}\in G$, the logical circuit
\[
\prod_{Q\in S} \; \overline{\mathsf C^{\, m-1}\mathsf Z \bigl[Q,\ \sigma_1(Q),\ \sigma_2(Q),\ \ldots,\ \sigma_{ m-1}(Q)\bigr]}
\]
can be performed in depth one via a circuit of physical $\mathsf C^{ m-1}\mathsf{Z}$ gates. 
\item In particular, given any orbit $L_i\subseteq L$ and any $m$-tuple of logical qudits on this orbit, $Q_0,\dots,Q_{m-1}\in L_i$, each belonging to a distinct code block of $\mathcal Q^{\otimes m}$, the logical gate $
\overline{\mathsf C^{\, m-1}\mathsf Z \bigl[Q_0,\ Q_1, \ldots,Q_{m-1}\bigr]} $ can be performed in depth one.
\end{enumerate}  
\end{theorem}

\begin{corollary}\label{Corollary depth}
Let $r$, $G$ and $\mathcal Q$ be defined as in Theorem~\ref{Theorem quasi-G-quantum}. For $r=1$, $\mathcal Q$ is a regular quantum $G$ code and is fully addressable via $
\overline{\mathsf C^{\, m-1}\mathsf Z}$  gates. In that case, the minimal physical depth of any logical interblock $\mathsf C^{ m-1}\mathsf{Z}$-circuit on $\mathcal Q^{\otimes n}$ is upper bounded by $k^{m-1}$.
\end{corollary}
\begin{remark}
Informally, the mutli-control-$Z$ gates of the code of Theorem \ref{Theorem quasi-G-quantum} can specifically address sets of logical qudits only orbit-wise, and gates involving qudits on different orbits can be parallelized as long as those gates involve the same group elements $\sigma_1,\dots,\sigma_{m-1}$.Hence, there is a trade-off between the dimension of the quantum code and the level of addressability: when the orbit size is sublinear in $n$, puncturing/shortening on more orbits increases the dimension while only puncturing a single orbit yields full addressability.
\end{remark}

We prove Theorem \ref{Theorem quasi-G-quantum} in several steps. First, we establish technical lemmas on the module structure of $\mathcal Q$. Then, we discuss how to physically implement the sequence of logical gates of item \textit{(3)}.\par
All the maps discussed in this section are defined in the category of $\F_q[G]$-modules unless otherwise stated. Since $C$ is a left $\F_q[G]$-module and $L$ is a union of orbits of the $G$-action, it follows that both $C_1$ and $C_0^\perp$ are left $\F_q[G]$-modules. Throughout, we let $\pi_L:C\to \F_q^L$ and $\pi_{\overline{L}}:C\to C_1$ denote the projection on the coordinates of $L$ and $\overline{L}$, respectively. Clearly, these maps are compatible with the $G$-action since $L$ is a union of orbits.

\begin{lemma}\label{lemma regular quantum code}
Let $C$, $C_0$, and $C_1$ be defined as in Theorem~\ref{Theorem quasi-G-quantum}. Then, $\pi_L$ is surjective, $\pi_{\overline{L}}$ is an isomorphism, $C_0^\perp \subseteq C_1$ and $C_1/C_0^\perp$ is a free $\F_q[G]$-module of rank $r$. In particular, there exist elements $g_1,\dots,g_r\in C$ such that $\pi_{L_i}(g_i) = (1,0,\dots,0)$ and $$C_1/C_0^\perp=\bigoplus_{i=1}^r\F_q[G]\cdot \pi\circ \pi_{\overline{L}}(g_i).$$
Moreover, the map \begin{equation}\label{equation section speciale}
    \begin{array}{cccc}
\fs:& C_1/C_0^\perp &  \longrightarrow & C_1 \\
  & \pi_{\overline{L}}(g_i) + C_0^\perp & \longmapsto & \pi_{\overline{L}}(g_i)\\
\end{array}
 \end{equation}
is a well defined section of $\pi$.
\end{lemma}
 
\begin{proof}
Since $d(C^\perp)>r\ell$, the projection $\pi_L$ is surjective. Indeed, $d(C^\perp)>r\ell$ implies that $(C^\perp)_{\overline L=0}=\{0\}$ since $|L|\leqslant r\ell$. Using Lemma \ref{Lemma puncturing}, one has $(C_{L})^\perp=\pi_L(C)^\perp=\{0\}$, hence $\pi_L(C) = \F_q^L$. Moreover, since $L$ is a union of $r$ orbits of $G$, the space $\F_q^L$ is a free $\F_q[G]$-module of rank $r$. Thus, the exact sequence
\[0 \longrightarrow \ker \pi_L \longrightarrow C \longrightarrow \F_q^L \longrightarrow 0\]
splits. Let $g_1,\dots,g_r\in C$ be such that $\pi_{L_i}(g_i) = (1,0,\ldots,0)$ (which exists by surjectivity), one has $C = \bigoplus_{i=1}^r\F_q[G]\cdot g_i \oplus \ker \pi_L$ and the first summand is clearly free. We are going to show that $\pi_{\overline{L}}$ is an isomorphism of $\F_q[G]$-modules. We only need to show its injectivity. Let $x\in \ker \pi_{\overline{L}}$. Since $C$ is self-orthogonal, one has for all $\sigma \in G$ and all $g\in \{g_1,\dots,g_r\}$,
\[\langle  u\star x,\sigma \cdot g\rangle = \langle u_L\star \pi_L(x),\sigma\cdot\pi_L(g) \rangle + \langle  u_{\overline{L}}\star\pi_{\overline{L}}(x),\sigma\cdot\pi_{\overline{L}}(g)\rangle = 0.\]
Since $\pi_{\overline{L}}(x) = 0$, one has $\langle u_L\star \pi_L(x),\sigma \cdot \pi_L(g) \rangle=0$. 
Then, as the permutations of the elements of $\{g_1,\dots,g_r\}$ form  the canonical basis of $\F_q^L$, one has $u_L\star \pi_L(x) = 0$, and $\pi_L(x) = 0$ as well, because $u$ has only non-zero components. Thus $x=0$, and 
\[C_1 = \pi_{\overline{L} }(C) = \left(\bigoplus_{i=1}^r\F_q[G]\cdot \pi_{\overline{L}}(g_i) \right)\oplus \pi_{\overline{L}}(\ker \pi_L) =\left(\bigoplus_{i=1}^r\F_q[G]\cdot \pi_{\overline{L}}(g_i) \right) \oplus C_0^\perp. \]
In particular, $C_0^\perp \subseteq C_1$ and $C_1/C_0^\perp$ is a free $\F_q[G]$-module of rank $r$ generated by the elements of $\{\pi_{\overline{L}}(g_i) + C_0^\perp:1\leqslant i\leqslant r\}$ and one has a section as given by Equation \eqref{equation section speciale}.
\end{proof}
From now on, we denote by $\pi:C_1\to C_1/C_0^\perp$ the canonical quotient map.
\begin{lemma}\label{lemma:multiplication-property decomposition}
Let $C$, $C_0$, $C_1$ and $u$ be defined as in Theorem~\ref{Theorem quasi-G-quantum}. Then, there exists an encoding $\Enc:\F_q^L\to C_1/C_0^\perp$ such that for any $m+1$ codewords $f^0,f^1, \dots, f^m \in C_1 $ it holds that
\begin{equation}
    \sum_{P\in \overline{L} } \left( u_{\overline{L}}\star f^0\star f^1\dots \star f^{m-1} \right)_P=\sum_{Q\in L} \left(u_L\star x^0\star x^1\star\dots\star x^{m-1}\right)_Q
\end{equation}
where $x^i \in \F_q^L$ is such that $\Enc (x^i) = \pi(f^i)$ for $0\leqslant i \leqslant m-1$.
\end{lemma}

\begin{proof}
The map $\Enc $ completes the following commutative diagram:
\begin{equation}\label{Diagramme commutatif}
        \begin{tikzcd}
		C \arrow[dd, "\pi_L"'] \arrow[rr, "\pi_{\overline{L}}"]             &  & C_1 \arrow[dd, "\pi"] \\
		&  &                       \\
		\F_q^L  \arrow[rr, "\sim","\Enc"']  &  & C_1/C_0^\perp.
	\end{tikzcd}
\end{equation}
As stated in Lemma \ref{lemma regular quantum code}, there exists $\{g_1,\dots g_r\}\in C$ such that $\pi_{L_i}(g_i)=(1,0,\dots,0)$. Then $\Enc$ is defined by $\pi_L(g_i)\mapsto \pi\circ \pi_{\overline{L}}(g_i)$, for $1\leqslant i\leqslant r$, and extended by linearity. Now, recall that by Lemma \ref{lemma regular quantum code}, $\pi_{\overline{L}}$ is an isomorphism. Hence, for $f^i\in  C_1$, let $\tilde f^i\in C$ denotes its antecedent by $\pi_{\overline{L}}$. From the $m$-multiplication property of $C$, 
\begin{align*}
    \sum_{Q\in  L } \left( u\star \tilde f^0\star \tilde f^1\dots \star \tilde f^{m-1}\right)_Q & +\sum_{P\in \overline{L} } \left( u\star \tilde f^0\star \tilde f^1\dots \star \tilde f^{m-1}\right)_P=0,
\end{align*}
or written differently, by projecting the vectors onto the components over which the sum are defined
\begin{align*}\label{Equation left-right}
    \sum_{Q\in L } \left( u_L\star \pi_L(\tilde f^0)\star \pi_L(\tilde f^1)\star\dots \star \pi_L(\tilde f^{m-1})\right)_Q=\sum_{P\in \overline{L} } \left(  u_{\overline{L}}\star  \pi_{\overline{L}} (\tilde f^0)\star \pi_{\overline{L}} (\tilde f^1)\dots \star \pi_{\overline{L}} (\tilde f^{m-1})\right)_P
\end{align*}
Now, from the diagram in Equation \eqref{Diagramme commutatif}, we know that $\Enc( \pi_{L}(\tilde f^i))= \pi(f_i)$. We can thus set $x^i = \pi_{L}(\tilde f^i)$. Moreover, $\pi_{\overline{L}} (\tilde f^i) = f^i$, proving the claim.
\end{proof}

We now introduce the following notion of \textit{modulation function}, which originates from \cite{guemard2025good}.

\begin{definition}\label{definition modulation function}
Let $C$ be the code defined in Theorem \ref{Theorem quasi-G-quantum}. Given any subset $S\subseteq L$, corresponding to a subset of logical qudits of $\mathcal Q$ via $\Enc$, and a vector $\lambda\in \F_q^{L}$, we define the associated \emph{modulation function}
\[
\Lambda(\lambda,S) \coloneqq \fs \circ \Enc(1_S\star\lambda\star u_L^{-1}),
\]
where we write\footnote{Recall that $u$ has only non-zero components.} $u^{-1}\coloneq (u_Q^{-1})_{Q\in N}$ and $1_S$ denotes the vector with components 1 on the positions of $S$ and $0$s on the complement of $S$.
\end{definition}

Note that $\Lambda \in C_1$, since it is defined via the section $\fs:C_1/C_0^\perp\to C_1$. Using this modulation function, we now exhibit a physical circuit whose effect at the logical level is to address specific tuples of logical qudits with multi-control-$Z$ gates.

\begin{lemma} \label{lemma: main circuit}
Let $\mathcal{Q}$ be the quantum code defined in Theorem \ref{Theorem quasi-G-quantum}. Let $S\subseteq L$ be a subset of indices, describing a subset of logical qudits. Then, for any logical state $\ket{\psi} \in \mathcal{Q}^{\otimes m}$, given any modulation function $\Lambda\coloneq \Lambda(\lambda,S)$ as in Definition \ref{definition modulation function}, and any group elements $\sigma_1,\sigma_2,\dots,\sigma_{m-1}\in G$, it holds that
\begin{align}\label{Equation physical circuit}
\prod_{Q\in S}\overline{\mathsf{C}^{m-1}\mathsf{Z}^{\lambda_Q}[Q,\sigma_1(Q),\sigma_2(Q),\dots,\sigma_{m-1}(Q)]} \ket{\psi}
=  
\prod_{P\in \overline{L}}\mathsf{C}^{m-1}\mathsf{Z}^{u_{P} \Lambda_P}[P,\sigma_1(P),\sigma_2(P),\dots,\sigma_{m-1}(P)] \ket{\psi}.
\end{align}
\end{lemma}

\begin{proof}
The proof is similar to that of \cite{guemard2025good}. By linearity, it is enough to determine the action of the sequence of gates in the r.h.s of Equation \eqref{Equation physical circuit} on a tensor product of logical basis states $\overline{\ket{x^0}}\:\overline{\ket{x^1}} \dots \overline{\ket{x^{m-1}}}$, where $x^j=(x^j_Q)_{Q\in L} \in \F_q^{L}$, and $ \overline{\ket{x^j}}\in \mathcal{Q}$ belongs to the $j$-th code block of $(\mathcal{Q})^{\otimes m}$. Let us recall that the encoding of $x\in \F_q^L$ is defined as $\overline{\ket{x}}\coloneq\ket{\Enc(x) }= \sum_{b\in C_0^\perp} \ket{\fs\circ \Enc(x)  +b}$, where $\fs$ is a section of $\pi$. By linearity, it is sufficient to check the action of the set of physical gates on a tensor product $\ket{\mathbf F}\coloneq\ket{f^0}\ket{f^1}\dots \ket{f^{m-1}},$ where $\ket{f^j}$ is any summand $\ket{\fs\circ\Enc( x^j)+b^j}$, for $b^j\in C_0^\perp$, of $\overline{\ket{x^j}}$, $j=0,\dots,m-1$. It follows that $f^j$ is any codeword of $C_1$. Moreover, let us recall that the set of physical qudits is in bijection with the set $\overline{L}=N\setminus L$. Hence, we index the component of a codeword $f^j\in C_1$ such that $f^j=(f^j_P)_{P\in \overline{L}}\in \F_q^{\overline{L}}$. Then, applying the circuit in the r.h.s of Equation \eqref{Equation physical circuit} to $\ket{\mathbf  F}$ yields
\begin{align*}
\prod_{P\in \overline{L}} \mathsf{C}^{m-1}\mathsf{Z}^{u_{P}   \Lambda_P}[P,\sigma_1(P),\sigma_2(P),\dots,\sigma_{m-1}(P)]\:  \ket{\mathsf F}
        =&   \prod_{P\in \overline{L}}\exp\left(  i\pi  \text{tr}\left( u_{P}   \Lambda_P   f^0_P    f^1_{\sigma_1(P)} \dots f^{m-1}_{\sigma_{m-1}(P)} \right) \right) \ket{\mathbf F} \\
        =& \exp \left(  i\pi  \text{tr} \left( \;  \sum_{P\in\overline{L} } u_{P}   \Lambda_P    f^0_Pf^1_{\sigma_1(P)} \dots f^{m-1}_{\sigma_{m-1}(P)} \right) \right)  \ket{\mathbf F}.
    \end{align*}
Given any codeword $f\in C_1$ and element $\sigma\in G$, we have that $f_{\sigma(P)}=(\sigma^{-1}\cdot f)_P,$ where the r.h.s is the $P$ component of the vector $\sigma^{-1}\cdot f$, which is still a codeword of $C_1$, since $C_1$ is a quasi-$G$ code. It follows that we may transform the sum in the exponential by applying Lemma \ref{lemma:multiplication-property decomposition} to the $m$ functions $\sigma_j^{-1}\cdot f^j$ for $0\leqslant j\leqslant m-1$, with the function $f^m$  taken to be $ \Lambda=\Enc (1_S\star\lambda\star u_L^{-1})$. This yields
\begin{align*}
   \sum_{P\in \overline{L} } u_{P} \:  \Lambda_P  \: f^0_P  f^1_{\sigma_1(P)} \dots f^{m-1}_{\sigma_{m-1}(P)}  
        =&\sum_{P\in \overline{L}} \left(  u_{\overline{L}}  \star  \Lambda \star     f^0    \star (\sigma_1^{-1}\cdot f^1)\star\dots \star  (\sigma_{m-1}^{-1}\cdot f^{m-1})\: \right)_P\\
        =&\sum_{Q\in L } \left( u_L\star 1_S\star\lambda\star u_L^{-1}  \star  x^0\star (\sigma_1^{-1}\cdot x^1) \star\dots \star (\sigma_{m-1}^{-1}\cdot x^{m-1}) \right)_Q \\
        =& \sum_{Q\in S}\lambda_Q   \: x^0_{Q}  \: x^1_{\sigma_1(Q)}\dots x^{m-1}_{\sigma_{m-1}(Q)},
    \end{align*}
where the first equality is a simple rewriting in terms of Schur product, the second equality follows from Lemma \ref{lemma:multiplication-property decomposition}, using that $$\sigma_j^{-1}\cdot f^j=\sigma_j^{-1}\cdot (\fs\circ\Enc(x^i)+b^i)=\fs\circ\Enc(\sigma_j^{-1}\cdot x^i)+\sigma_j^{-1}\cdot b^j,$$
for any $\sigma_j \in G, j=0,\dots,m-1$ and hence $\pi(\sigma_j^{-1}\cdot f^j)=\Enc(\sigma_j^{-1}\cdot x^i)$. In the last line we have used the vector $1_S$ to transform a sum over $L$ into a sum over $S$. It follows directly that
 \begin{align}\label{Equation effect of circuit on arbitrary vectors}
       \prod_{P\in \overline{L}}\mathsf{C}^{m-1}&\mathsf{Z}^{u_{P}   \Lambda_P}[P,\sigma_1(P),\sigma_2(P),\dots,\sigma_{m-1}(P)] \:\ket{\mathbf F}   =\prod_{Q\in S}\exp \left(  i\pi \:\text{tr}\left( \lambda_Q   x^0_{Q}   x^1_{\sigma_1(Q)}\dots x^{m-1}_{\sigma_{m-1}(Q)}\right)\right)\ket{\mathbf F}.
\end{align}
Since this holds for any vector $f^0,f^1,\dots,f^{m-1}\in C_1$, by linearity it also holds for any product  $\overline{\ket{\boldsymbol{x}}}\coloneq\overline{\ket{x^0}}\:\overline{\ket{x^1}}\dots\overline{\ket{x^{m-1}}}$ of logical basis states. We thus obtain

\begin{align*}
       \prod_{P\in \overline{L}}\mathsf{C}^{m-1}\mathsf{Z}^{u_{P}   \Lambda_P}[P,\sigma_1(P),\sigma_2(P),\dots,\sigma_{m-1}(P)] \overline{\ket{\boldsymbol{x}}}
       &=\prod_{Q\in S}\exp \left(  i\pi\:\text{tr}\left( \lambda_Q   x^0_{Q}   x^1_{\sigma_1(Q)}\dots x^{m-1}_{\sigma_{m-1}(Q)}\right)\right)\overline{\ket{\boldsymbol{x}}}\\
    &= \prod_{Q\in S}\overline{\mathsf{C}^{m-1}\mathsf{Z}^{\gamma_Q}[Q,\sigma_1(Q),\sigma_2(Q),\dots,\sigma_{m-1}(Q)]}\:\overline{\ket{\boldsymbol{x}}}.
\end{align*}
By linearity, this holds for any other state $\ket{\psi}\in\mathcal Q^{\otimes m}$.
\end{proof}

We now have all the components required to prove Theorem \ref{Theorem quasi-G-quantum}.

\begin{proof}[Proof of Theorem \ref{Theorem quasi-G-quantum}]
For item\textit{ (1),} by construction it is clear that $C_0^\perp \subseteq C_1$. Moreover, since $L$ is a union of orbits of $G$, there is a well defined action of $G$ on $\overline{L}$ defining a subgroup of the automorphism group of $C_1$ and $C_0^\perp$, making them both left $G$-module. Moreover, by Lemma \ref{lemma regular quantum code}, $C_1/C_0^\perp$ is a free $G$-module, and it follows that $\mathcal Q$ is, by definition, a free quantum $G$ code.\par

We now treat the parameters given in item \textit{(2)}. For the dimension of $\mathcal Q$, we appeal to item \textit{(1)}, or equivalently Lemma \ref{lemma regular quantum code}, telling that $\mathcal Q$ is a free quantum $G$ code, and consequently $k_{\mathcal Q}=\dim (C_1/C_0^\perp)=\dim_{\F_q}\bigoplus_{i=1}^r\F_q[G]=r\ell$. For the distance, we have
\begin{align*}
 d(Q)\geqslant \min\{d(C_0), d(C_1)\} &= \min\big\{d\big((C_{L=0})^{\perp}\big),d(C_{\overline{L}})\big\} \\
&= \min\{d((C^{\perp })_{\overline{L}}),d(C_{\overline{L}})\}\\
&\geqslant \min\{d(C^{\perp}), d(C)\} - r\ell,
\end{align*}
where the second line follows from Lemma \ref{Lemma puncturing}.\par

For item\textit{ (3)}, in Lemma \ref{lemma: main circuit} we showed that the product of physical gates $$\mathcal C\coloneq\prod_{P\in \overline{L}}\mathsf{C}^{m-1}\mathsf{Z}^{u_{P} \Lambda_P}[P,\sigma_1(P),\sigma_2(P),\dots,\sigma_{m-1}(P)] $$
implements the logical circuit $\prod_{Q\in S}\overline{\mathsf{C}^{m-1}\mathsf{Z}^{\lambda_Q}[Q,\sigma_1(Q),\sigma_2(Q),\dots,\sigma_{m-1}(Q)]}$. First, for each multi-control-$Z$ gate of $\mathcal C$, the $m-1$  targeted physical qudits belong to different codeblocks. Suppose that two of these gates overlap on qudits $Q_1,Q_2\in S$, say in the $i$th codeblock, for some $1\leqslant i\leqslant m$. Then, we must have $\sigma_i(Q_1)=\sigma_i(Q_2)$, which means that $Q_1=Q_2$ since $G$ acts freely on $L$. Hence $\mathcal C$ is a sequence of non-overlapping physical gates and can be implemented transversally. 

For the last claim, item \textit{(4)}, we may take the singleton $S=\{Q_0\}$. Then, given any $(m-1)$-tuple of logical qudits $Q_1,\dots,Q_{m-1}\in L_i$, we can find elements $\sigma_1,\ldots,\sigma_{m-1}\in G$ such that $Q_i=\sigma_i(Q_0)$, for $1\leqslant i\leqslant m-1$, since the code is a free quantum $G$ code. Hence, the logical $
\overline{\mathsf C^{\, m-1}\mathsf Z \bigl[Q_0,\ Q_1, \ldots,Q_{m-1}\bigr]} $ targets any $m$-tuple of logical qudits and is also performed in depth one via the circuit of Lemma \ref{lemma: main circuit}.
\end{proof}

\begin{proof}[Proof of Corollary \ref{Corollary depth}]
The proof follows the same line as that of Section 3.3 in \cite{guemard2025good}. Since we have a single orbit, the multi-control-$Z$ gates can target any $m$-tuples of logical qudits. Successively applying the physical circuit of Lemma \ref{lemma: main circuit} for all possible choices of $\sigma_1,\dots,\sigma_{m-1}\in G$, we can obtain the all-to-all inter-block $\overline{\mathsf C^{m-1}\mathsf Z}$ circuit (with a choice of modulation $\lambda$ for each gate), with physical depth $|G|^{m-1}=k^{m-1}$. Any other circuit can be obtained from the all-to-all circuit by appropriately choosing the sets $S$ at each time step.
\end{proof}

\subsection{Known examples}

Instances of codes that share similar properties to that of quantum group codes are given in \cite{he2025addressable,he2025good,guemard2025good}. Before we introduce our main construction in the next section, we connect these known examples to the formalism developed in Theorem \ref{Theorem quasi-G-quantum}.
    
    First, in \cite{he2025good}, the code $C$ is an algebraic geometry code $C_{\mathcal L}(P,D)$ with the $4$-multiplication property taken from \cite{Stichtenoth2006} and defined over a Galois extension field $F'$ of some fixed function field $F$. The set $N$ is the set of rational places in the support of the divisor $P$. The divisor $D$ is another divisor disjoint from $P$.  The group $G$ is the Galois group $\Gal (F'/F)$ and the set $L$ is the set of places of $F'$ above a given rational place $P_0$ of $F$. The set $N$ is the set of all rational places of $F'$ that lie above rational places of $F$.
    
    Second, in \cite{guemard2025good}, similarly to \cite{he2025good}, the code $C$ is an instantiation of the construction of \cite{Stichtenoth2006} that has the $m$-multiplication property for any fixed $m\geqslant 2$. The group $G$ is the Galois group $\Gal( F/F'')$ of any intermediate field extension $F''$, where $F\subseteq F''\subseteq F'$. This allows for control over the parameters of the quantum code. Then $L$ is the set of places of $F'$ above a given rational place of $F''$. It was first proved in \cite{guemard2025good} that the gates of \cite{he2025good} can be parallelized, allowing logical multi-control-$Z$ circuits to be implemented with a minor space-time overhead.

Finally, in \cite{he2025addressable}, the constructed quantum code with addressable gates are quotients of quantum $G$ codes. There construction resembles the ones from \cite{he2025good,guemard2025good} and Section~\ref{Good quantum codes with addressable and parallelizable}. In that case, the code $C$ is a Reed-Solomon code over $\F_q$ with a restricted subset of evaluation points $N$ that we now define. Let $F$ be a subfield of $\F_q$. Then, $F$ acts on $\F_q$ by translation, and hence on the coordinates of the Reed-Solomon code. We may set $G\coloneq F$ and find a subset of $\F_q$ stabilized by $F$ and over which $F$ acts freely: $F$ is such a subset. Then, instead of $L$ being a full orbit of $G$, it is defined as a subset of the orbit of some $\zeta\in \F_q\setminus F$. In \cite{he2025addressable}, those choices are made so as to ensure the multiplication-property of $C$. Then, following the notation of Theorem \ref{Theorem quasi-G-quantum}, the set of indices $N\coloneq L\cup F$ and it indeed holds that $\overline{L}=F$ is stabilized by $G$ and disjoint from $L$.

\section{Good quantum group codes with enhanced decoding }\label{Section: good quantum}

\subsection{Good quasi-abelian AG codes with the $m$-multiplication property}\label{Section Good quasi-abelian AG codes}

In this section, we use a construction from~\cite{CouGas26} that produces quasi-$G$ codes from algebraic geometry, where $G$ is an abelian group of large order relative to the length of the code. Since the original paper is written in a geometric language, we explain some of the ideas using the language of function fields. We show that this construction preserves the $m$-multiplication property, and is therefore suitable for designing quantum codes with transversal gates that are addressable. The authors of~\cite{CouGas26} also showed that the traditional decoding algorithm for AG codes can be adapted to benefit from the additional module structure of these codes, making them efficiently encodable and decodable.

The aim of this section is to prove the following theorem:

\begin{theorem}\label{theo-lift-AG-code}
Let $q>4$ be a prime power. Let $F$ be a function field of genus $g>1$ with field of constants $\F_q$, such that $\frac{N(F)}{g} > 2$, where $N(F)$ is the number of rational places of $F$. Let $P$ be the divisor sum of $n$ rational places of $F$, with $n>2g$ and $D$ be an effective divisor with $\operatorname{supp} D \cap \operatorname{supp} P = \emptyset $ satisfying
\[2g-2 < \deg D < \deg P.\]
Let $C \coloneq C_\cL(P,D)$, a code of length $n$ and dimension $k$. Assume there exists an integer $m$ and $\eta \in \Omega_F$ with a simple pole at every rational place of $F$ such that
\[(m+1)D\leqslant P+(\eta),\]
hence $C$ satisfies the $m$-multiplication property.

Then, there exists an AG code $\tC=C_\cL(\widetilde P,\widetilde D)$ over $\F_{q^2}$ with the same rate and designed relative distance as $C$, which is a quasi-$G$ code for $G$ an abelian group such that
\[\alpha_q n \leqslant \log |G| \leqslant \beta_q n,\]
where $\alpha_q$ and $\beta_q$ are two positive constants depending solely on $q$. Moreover, $\tC$ is a free $\F_{q^2}[G]$-module of rank $k$ and satisfies the $m$-multiplication property.
The dual code $\tC^\perp$ has the same rate and designed relative distance as $C^\perp$, and is a quasi-$G$ code as well. It is also free as a $\F_{q^2}[G]$-module, of rank $n-k$.
\end{theorem}

Theorem~\ref{theo-lift-AG-code} yields many interesting implications when applied to a good (or excellent) family of AG codes with the multiplication property.
For instance, let $(C_i)_{i\geqslant 0}$ be the family of AG codes from~\cite{Wills2025}. Applying Theorem~\ref{theo-lift-AG-code}, one obtains a family of good quasi-abelian AG codes $(\tC_i)_{i\geqslant 0}$
with large permutation automorphism group satisfying the $m$-multiplication property for any positive integer $m<6$.
Using Proposition~\ref{prop-encoding-quasi-G-codes}, because these codes are free $\F_{q^2}[G]$-modules, it follows that they can also be encoded in quasi-linear time in their length. It is also shown in~\cite{CouGas26} that they can be decoded in quasi-quadratic time (see~\cite[Section 9.3]{CouGas26}).

\subsubsection{Class field theory}

Our construction requires an unramified abelian extension of function field $\tF/F$. Class field theory gives a complete description of all abelian extensions of a function field $F/\F_q$ in terms of its Picard group $\Pic(F)$. We will make a summary of a few results that suit our needs.

Let $P$ be a rational place of $F$. The following classical proposition \cite{Rosen87}, or see ~\cite[Section~3.3.2, Proposition~16 and~17]{GasnierPhD}, gives a description of unramified abelian extensions of $F$ in which $P$ splits completely.

\begin{proposition}\label{prop-CFT-Hilbert}
Fix a separable closure $F_{\mathrm{sep}}$ of $F$. There exists a maximum unramified abelian extension of $F$ in $F_{\mathrm{sep}}$ in which $P$ splits completely. It is called the \textit{Hilbert class field} of $F$ (with respect to $P$) and is denoted by $\Hil_{P}(F)$. The field of constants of $\Hil_{P}(F)$ is the same as the one of $F$.
\end{proposition}
 
The following theorem~\cite[Theorem~1.3]{Rosen87} gives a description of the Galois group of $\Hil_{P}(F)$, classifies intermediate extensions between $F$ and $\Hil_{P}(F)$, and gives a criterion for rational places of $F$ to split completely in any intermediate extension.

\begin{theorem}\label{theo-Artin-and-consequences}
With the notation above, one has
\[ \Gal(\Hil_{P}(F)/F)\simeq \Pic^0(F). \]
In particular, Galois theory states that there is a one-to-one correspondence between subgroups $H$ of $\Pic^0(F)$ and unramified abelian extensions of $F$ in ${F}_{sep}$, denoted $\tF_H$, in which $P$ splits completely. The Galois group of $\tF_H$ is naturally isomorphic to $\Pic^0(F)/H$. 

In addition, a rational place $P'$ of $F$ splits completely in $\tF_H$ if and only if the equivalence class of the degree zero divisor $P'-P$ in $\Pic^0(F) $ belongs to $H$.
\end{theorem}

This proposition will be later used to lift algebraic geometry codes defined over some function field $F$. It remains to select appropriately the function field $F$ and the subgroup $H$. This is the aim of the next section.

\subsubsection{Lifting function fields with many rational places}

Let $F/\F_q$ be a function field over $\F_q$ of genus $g(F)$ with field of constants $\F_q$. Let $P$ be the divisor sum of $n$ rational places of $F$ and $D$ be a divisor such that $\operatorname{supp}(P)\cap\operatorname{supp}(D)=\emptyset$ and
\[
2g(F)-2<\deg D <n.
\]
Let $m$ be a positive integer and $\eta\in\Omega_F$ be a differential with a simple pole at every rational place of $P$, and assume $(m+1)D\leqslant P+(\eta)$, so that $C_\cL(P,D)$ satisfies the $m$-multiplication property.

\begin{theorem}\label{theo-lift-mult-prop}
Let $F,P,D,\eta$ be as above. Let $F' = \F_{q^2}\otimes_{\F_q} F$- be the constant field extension of $F$ to $\F_{q^2}$. Then, there exists a separable unramified extension $\widetilde{F}_H/F$ with the following properties:
\begin{enumerate}[(a)]
    \item The field of constants of $\widetilde{F}_H$ is $\F_{q^2}$.
    \item  The extension $\widetilde{F}_H/F'$ is abelian and $G\coloneq \Gal (\widetilde{F}_H/F')$ is an abelian group of order \[|G| \geqslant (\sqrt{q}-1)^{2g(F)}.\]
\end{enumerate}
Moreover, there exists divisors $\widetilde{P},\widetilde{D}\in\operatorname{Div}(\widetilde{F}_H)$ such that:
\begin{enumerate}[(1)]
    \item $\widetilde{P}$ is the sum of the $n\cdot|G|$ rational places (over $\F_{q^2}$) above the places of $P$.
    \item $\deg \widetilde{D} = \deg D \cdot |G|$ and $2g(\widetilde{F}_H)-2<\deg \widetilde{D}<n\cdot|G|.$
    \item $\widetilde{P}$ and $\widetilde{D}$ are invariant under the action of $G$.
    \item The Riemann-Roch space $\cL(\widetilde{D})$ is a free $\F_{q^2}[G]$-module of rank $\deg(D)-g(F)+1$.
    \item $ \operatorname{Cotr}_{\widetilde{F}_H/F}(\eta)$ has a simple pole at every place of $\widetilde P$.
    \item The following inequality holds: \[(m+1)\widetilde{D}\leqslant\widetilde{P}+\left(\operatorname{Cotr}_{\widetilde{F}_H/F}(\eta)\right).\]
\end{enumerate}
\end{theorem}

\begin{proof} We start with items \textit{(a)-(b)}. First, $F'$ is a function field over $\F_{q^2}$ of genus $g(F)$, and $F'/F$ is an unramified extension of degree $2$. The rational places of $F'$ (whose residue field is $\F_{q^2}$) are composed of:
\begin{itemize}
  \item the places lying above rational places $Q$ of $F$. In this case, $r_Q=e_Q=1$ and $f_Q=2$.
  \item the places lying above places $Q'$ of degree $2$ of $F$. In this case, $e_{Q'}=2$ and $r_{Q'}=f_{Q'}=1$.
\end{itemize}
Moreover, the conorm map induces an injective map from $\operatorname{Pic}^0(F)$ to $\operatorname{Pic}^0(F')$, so we may identify $\operatorname{Pic}^0(F)$ with a subgroup $H$ of $\operatorname{Pic}^0(F')$.\footnote{With the geometric point of view, this is just saying that $\F_q$-rational points of the Jacobian variety form a subgroup of the $\F_{q^2}$-rational points.}
Let $P_0$ be a rational place of $F'$ lying above a rational place of $F$. We may define $\widetilde{F}_H$ as in Theorem~\ref{theo-Artin-and-consequences}, the unramified abelian extension of $F'$ in which $P$ splits completely with Galois group isomorphic to $\Pic^0(F')/H$. Namely, we have a tower of Galois extensions
\[
F \subseteq F'\subseteq \widetilde{F}_H \subseteq \mathrm{Hil}_{P_0}(F').
\]
Since $F'/F$ is unramified and $\widetilde{F}_H/F'$ is unramified then the extension $\widetilde{F}_H/F$ is unramified as well.

A theorem of Weil~\cite{Weil48}, proving the Riemann Hypothesis for function fields of transcendance degree $1$ over a finite field implies that
\[
(\sqrt{q}-1)^{2g(F)}\leqslant |\Pic^0(F)| \leqslant (\sqrt{q}+1)^{2g(F)} \text{ and } ({q}-1)^{2g(F')}\leqslant |\Pic^0(F')| \leqslant({q}+1)^{2g(F')}
\]
and since $g(F')=g(F)$, the order of the Galois group satisfies 
\[
|\operatorname{Gal}(\widetilde{F}_H/F')|=|\Pic^0(F')|/|\Pic^0(F)|\geqslant(\sqrt{q}-1)^{2 g(F)}.
\]

Concerning items \textit{(1)-(4)}, for every rational place $P'$ of $F'$, we have $P'-P_0\in H$ if and only if $P'$ lies above a rational place of $F$, in which case $P'$ splits completely in $\widetilde{F}_H$. Since $[\widetilde{F}_H:F']=|G|$, there are $n\cdot |G|$ rational places above places of $\operatorname{supp}(P)$.
Let $\widetilde{P} = \Con_{\tF/F}(P)$, and $\widetilde{D} = \Con_{\tF/F}(D)$. As conorms, they are invariant under the action of $G$. The freeness of $\cL(\tilde D)$ as an $\F_{q^2}[G]$-module is proved in~\cite[Section 4, Theorem 1]{CouGas26}.
Then, the operator $\Con_{\widetilde{F}_H/F}$ multiplies the degree of the divisors by $\frac{[\widetilde{F}_H:F]}{[\F_{q^2}:\F_q]} = |G|$, and one has 
\[
2g(\tF)-2 = |G|\cdot(2g(F)-2)
\]
by the Riemann-Hurwitz formula~\cite[3.4.13]{Stichtenothbook}, since the extension $\widetilde{F}_H/F$ is unramified. 

Finally, for item \textit{(5)-(6)}, since $\widetilde{F}_H/F$ is unramified, one has
$\left(\operatorname{Cotr}_{\widetilde{F}_H/F}(\eta)\right) = \Con_{\widetilde{F}_H/F}((\eta))$. Thus, $\operatorname{Cotr}_{\widetilde{F}_H/F}(\eta)$ has simple poles at the places above $P$. Finally one has
\[(m+1)\Con_{\widetilde{F}_H/F}(D) \leqslant \Con_{\widetilde{F}_H/F}(P)+\Con_{\widetilde{F}_H/F}((\eta)). \]
\end{proof}

\subsubsection{Good quasi-abelian AG codes} 

We now examine the implications of Theorem \ref{theo-lift-mult-prop} in the context of AG codes.

\begin{corollary}\label{cor-AG-codes}
Let $F$, $P$, $D$, $\eta$ be as in Theorem~\ref{theo-lift-mult-prop}. Let $C\coloneq C_\cL(P,D)$ be an AG code with parameters $[n,k,d]$. In particular, since $(m+1)D\leqslant P+(\eta)$, $C$ has the $m$-multiplication property. Then, there exists an unramified field extension $\tF/F$ and a linear code $\widetilde{C}$ over $\F_{q^2}$ with the following properties:
\begin{enumerate}[(1)]
    \item $\widetilde{C}\coloneq C_\cL(\widetilde{P},\widetilde{D})$, where $\widetilde P$ and $\widetilde D$ are conorms of $P$ and $D$ in $\tF$, respectively.
    \item $\widetilde{C}$ is a quasi-$G$ code with $G$ abelian and $|G|\geqslant(\sqrt{q}-1)^{2g(F)}$.
    \item $\tC$ is a free $\F_{q^2}[G]$-module of rank $k$.
    \item $\widetilde{C}$ has the $m$-multiplication property.
    \item  $\widetilde{C}$ has parameters $[\widetilde{n},\widetilde{k},\widetilde{d}]$ with:
\[
\widetilde{k}=k\cdot|G| ,
\qquad
\widetilde{d}\geqslant (n-\deg  D)\cdot |G|,
\qquad
\widetilde{n}=n\cdot|G|.
\]
  \item The dual code $\tC^\perp$ is a free $\F_{q^2}[G]$-module of rank $(n-k)$, and has parameters $[\widetilde{n}^\perp,\widetilde{k}^\perp,\widetilde{d}^\perp]$ where :
  \[
\widetilde{k}^\perp=(n-k)\cdot|G| ,
\qquad
\widetilde{d}^\perp\geqslant (\deg  D-2g(F)+2)\cdot |G|,
\qquad
\widetilde{n}=n\cdot|G|.
\]
\end{enumerate}
\end{corollary}
\begin{proof}
Since one has $k = \deg(D)-g(F)+1$ by the Riemann-Roch theorem and $2g(\tF)-2 = |G|\cdot (2g(F)-2)$ by the Riemann-Hurwitz formula~\cite[3.4.13]{Stichtenothbook}, items \textit{(1)-(5)} follow from Theorem~\ref{theo-lift-mult-prop}.

For item \textit{(6)} recall that the dual code of $\tC = C_\cL(\widetilde{P},\widetilde{D})$ has a geometric nature: $\tC^\perp$ is the image of the (injective) residue map
\[
\begin{array}{cccc}
\operatorname{res}_{\widetilde{D}-\widetilde{P},\widetilde{P}}:
& \Omega(\widetilde{D}-\widetilde{P}) & \longrightarrow  & 
\F_{q^2}^{\widetilde{n}} \\
&  \omega & \longmapsto  & (\operatorname{Res}_{P_i}(\omega))_{i\leqslant \widetilde{n}}\\
\end{array}
\]
where $\Omega(\widetilde{D}-\widetilde{P}) \coloneq \{\omega \in \Omega_{\tF} \mid (\omega) \geqslant \widetilde{D}-\widetilde{P}\}$. The parameters of the code can be checked in~\cite[Theorem 2.2.7]{Stichtenothbook}, and the freeness in~\cite[Proposition 7]{CouGas26}.
\end{proof}

\begin{remark}\label{rem-G-order}
In~\cite{CouGas26}, the decoding algorithm requires, for technical reasons, $G$ to have order either a power of $p$ or coprime to $p$, where $p$ is the characteristic of $\F_q$. Any abelian group $G$ can be decomposed as $G = G_p \times G_\ell$ where $G_p$ is a $p$-group and $G_\ell$ has order coprime to $p$. Let $\tF'$ be the subfield of $\tF$ fixed by the smaller these two subgroups. One can replace $\tF$ and $\tC$ by $\tF'$ and $\tC'$, the lift of $\C_\cL(P,D)$ in $\tF'$, in the previous corollary. With these changes, every statement remains true except that $G'$, the Galois group of $\tF'$, only satisfies $|G'| \geqslant (\sqrt{q}-1)^{g(F)}$ (see \cite[Section 9.2]{CouGas26}). Then $|G'|$ is either a power of $p$ or coprime to $p$.
\end{remark}

We are now ready to prove Theorem~\ref{theo-lift-AG-code}.

\begin{proof}[Proof of Theorem~\ref{theo-lift-AG-code}]
    The assumptions of the Theorem allow to use Corollary~\ref{cor-AG-codes}: there exists an AG code $\tC$ which is a quasi-$G$ code with $G$ abelian. According to Remark~\ref{rem-G-order}, we can assume that $G$ has order either a power of the characteristic $p$ or coprime to $p$, if we only assume that $|G| \geqslant (\sqrt{q}-1)^{g}$. Then $\tC$ satisfies most of the conditions: we only need to show that there exists $\alpha_q$ and $\beta_q$, two positive constants depending only on $q$, such that
\[\alpha_q n \leqslant \log |G| \leqslant \beta_q n.\]
    
    We assumed that $n > 2g(F)$. Moreover, the Hasse-Weil bound~\cite{Weil48} states that \[n\leqslant N(F)\leqslant q+1+2g\sqrt{q}.\]
    Since $g>1$, one has
    \[ 2g < n \leqslant q+1+2g\sqrt{q} < (q+1+2\sqrt{q}) g, \]
    or equivalently 
\[\frac{n}{q+1+2\sqrt{q}} < g < \frac{n}{2}.\]
    Finally, according to the construction in Theorem~\ref{theo-lift-mult-prop} and~\cite{Weil48}, one has
    \[ (\sqrt{q}-1)^{g} \leqslant |G| \leqslant \left(\frac{q+1}{\sqrt{q}-1}\right)^{2g}, \]
    and hence
    \[\frac{\log(\sqrt{q}-1)}{q+1+2\sqrt{q}}\cdot n < \log(\sqrt{q}-1)\cdot g \leqslant \log |G| \leqslant 2\log\left(\frac{q+1}{\sqrt{q}-1}\right)\cdot g < \log \left(\frac{q+1}{\sqrt{q}-1}\right)\cdot n .\]
\end{proof}

In~\cite{CouGas26}, the authors show that the traditional cubic-time decoding algorithm of AG codes can be adapted to take advantage of the structure of free $\F_q[G]$-module. They design a randomized decoding algorithm with quasi-quadratic complexity in $|G|$, and since in Theorem~\ref{theo-lift-AG-code} one has $ n = \Theta(\log |G|)$, this yields quasi-quadratic complexity in the length of the code. We give the precise result in the following theorem.

\begin{theorem}[from \protect{\cite[Corollary~53.1]{GasnierPhD}}]\label{Theorem decoder}
We use the notation of Theorem~\ref{theo-lift-AG-code} and denote by $\tilde{n}$ the length of $\tC$, by $\tilde{g}$ the genus of the function field over which it is defined, and by $d^* = n-\deg(D)$ the designed distance of $C$. There exists a constant $\mathcal{C}_q$ such that the following holds: there exists a randomized (Las Vegas) decoding algorithm for $\tC$ allowing to decode up to $ |G|\cdot\lfloor(d^*+g-1)/2\rfloor-\tilde{g} $ errors with at most $\mathcal{C}_q\cdot \tilde{n}^2 \operatorname{polylog}(\tilde{n})$ operations.
\end{theorem}

Note that $ |G|\cdot\lfloor(d^*+g-1)/2\rfloor-\tilde{g}$ is equal to $\lfloor (\tilde d^*-\tilde{g}-1)/2 \rfloor$ when $(d^*+g-1)/2$ is an integer (where $\tilde d^*$ denotes the designed distance of $\tC$).

\subsection{Good quantum group codes from lifted AG codes}\label{Good quantum codes with addressable and parallelizable}

In this section we prove the main results of our work, Theorem \ref{Theorem, main} and Corollary \ref{Corollary MSD}. Our proof combines the framework for parallel addressability developed in Section \ref{section: Sufficient conditions for parallelizable addressability} and the AG code lifting method exposed in Section \ref{Section Good quasi-abelian AG codes}.

\begin{theorem}[Restatement of Theorem \ref{Theorem, main}]\label{Theorem two Good quasi-G-quantum}
Let $t>4$ be a power of two, and set $q =t^4$. Then, for any $1<m\leqslant \frac{t}{2}-1$ there exists:

\begin{enumerate}[(1)]
    \item a sequence of $q$-dimensional qudit regular quantum group codes with asymptotic parameters 
\[ 
[[n,\Theta (n/\log n) ,\Theta(n)]]_{q},
\]
and admitting transversal $\mathsf C^{\tilde m-1} \mathsf Z$-gates that are fully addressable for all $1<\tilde m\leqslant m$. Moreover, the depth of any logical circuit of $\mathsf C^{\tilde m-1} \mathsf Z$-gates acting across the $m$ code blocks of $\mathcal{Q}^{\otimes \widetilde m}$ is upper bounded by $O(k^{\tilde m-1})$, where $k$ denotes the dimension of $\mathcal Q$.

\item a sequence of $q$-dimensional qudit free quantum group codes with asymptotic parameters 
\[ 
[[n,\Theta (n) ,\Theta(n)]]_{q},
\]
and admitting a transversal $\mathsf C^{\tilde m }\mathsf Z$-gate as well as orbit-wise addressable transversal $\mathsf C^{\tilde m-1 }\mathsf Z$ gates across $\Theta(\log(n))$ orbits of size $\Theta(n/\log(n))$ for $1<\tilde m \leqslant m$.
\end{enumerate}

In both cases, the quantum code family admits a $O(n^2\cdot \operatorname{polylog}(n))$-time decoder\footnote{The polylogarithmic factor of the decoder is $\log^4(n)$} with a linear decoding radius.
\end{theorem}
\begin{proof}
Let $\{C^{(i)}: i\in \NN\}$ be an asymptotically good family of AG $q'$-ary code, where each $C^{(i)}\subseteq \F_{q'}^{n^{(i)}}$ has parameters $[n^{(i)},k^{(i)},d^{(i)}]_{q'}$, rate $\kappa^{(i)}$, relative distance $\delta^{(i)}$, and satisfies the assumption of Theorem \ref{theo-lift-mult-prop} for all $i\in \NN$. In particular, $C^{(i)}$ has the $m$-multiplication property and the $2$-multiplication property. Because the codes contain the all-one vector, they also have the $\tilde m$-multiplication property for every $1<\tilde m\leqslant m$. Suppose moreover that the dual-code family $\{(C^{(i)})^\perp: i\in \NN\}$ is asymptotically good as well. Instances of such codes can be found in \cite{golowich2024asymptoticallygoodquantumcodes,nguyen2024goodbinaryquantumcodes,Wills2025,he2025good,guemard2025good}. In the even-characteristic case of \cite{guemard2025good} this is especially verified for $t>4$,  $q'=t^2$ and $1<m\leqslant \frac{t}{2}-1$.\par

Then, following Theorem \ref{theo-lift-mult-prop}, each $C^{(i)}$, can be lifted into a $q$-ary code $\tilde C^{(i)}$, where $q\coloneq q'^2$, with parameters $[\tilde n^{(i)},\tilde k^{(i)},\tilde d^{(i)}]_{q}$ such that $\tilde C^{(i)}$ and $(\tilde C^{(i)})^\perp$ have the same rate $\kappa^{(i)},\kappa^{(i)\perp}$ and relative distance $\delta^{(i)},\delta^{(i)\perp}$ as those of $C^{(i)}$ and $(C^{(i)})^\perp$, respectively. That is, $\{\tilde C^{(i)}: i\in \NN\}$ and $\{(\tilde C^{(i)})^\perp: i\in \NN\}$ are asymptotically good families of codes. Moreover, these two codes are quasi-$G^{(i)}$ codes for some group $G^{(i)}$ of order $\ell^{(i)}$ satisfying $\alpha_q n^{(i)}\leqslant \log \ell^{(i)}\leqslant \beta_q n^{(i)}$ where $\alpha_q,\beta_q$ are positive universal constants.\par
Let $r^{(i)}\geqslant 1$ be an integer and let $L^{(i)}$ be a union of $r^{(i)}$ orbits of the action of $G^{(i)}$. Notice that since $\log\ell^{(i)}<\beta_qn^{(i)}$ and the family of lifted codes and their duals are both good, we can choose $r^{(i)}$ such that $r^{(i)}\ell^{(i)} < \min\{\tilde d^{(i)\perp},\tilde k^{(i)}\}$. Then, the code $\tilde C^{(i)}$ satisfies the assumptions of Theorem \ref{Theorem quasi-G-quantum}. We define the shortened code $C_0^{(i)}=C^{(i)}_{L^{(i)}=0}$ and the punctured code $C_1^{(i)}=C^{(i)}_{\overline{L}^{(i)}}$. It follows that the quantum code $\mathcal Q^{(i)}\coloneq \operatorname{CSS}(C_0^{(i)},C_1^{(i)})$ is a free quantum $G^{(i)}$ code with parameters $[[n^{(i)}_{\mathcal Q},k^{(i)}_{\mathcal Q},d^{(i)}_{\mathcal Q}]]$ such that
\[
\begin{aligned}
n^{(i)}_{\mathcal Q} &= \ell^{(i)}(n ^{(i)}-r^{(i)})\leqslant \ell^{(i)}\left( \frac{\log \ell^{(i)}}{\alpha_q}-r^{(i)}\right)\\
k^{(i)}_{\mathcal Q} &= r^{(i)}\ell^{(i)} \\
d^{(i)}_{\mathcal Q}  &\geqslant \min\{d((\tilde C^{(i)})^{\perp}), d(\tilde C^{(i)})\} - r\ell^{(i)}\\
      &\geqslant  \ell^{(i)}\left( \frac{ \min\{\delta^{(i)},\delta^{(i)\perp}\}\cdot\log \ell^{(i)}}{\beta_q} -r^{(i)}\right).
\end{aligned}
\]

We distinguish two cases.
\begin{enumerate}[(1)]
    \item First, suppose that $r^{(i)}=1$. Then, the family $\{\mathcal Q^{(i)}: i\in \NN\}$ has asymptotic parameters scaling as $[[n,\Theta(n/\log n), \Theta(n)]]$. From Theorem \ref{Theorem quasi-G-quantum}, each member $\mathcal Q^{(i)}$ is a regular quantum $G^{(i)}$ code which admits an encoder $\Enc^{(i)}$ allowing for the transversal implementation of the full set of addressable logical $\mathsf C^{\widetilde m-1}\mathsf Z$ gates for $1< \widetilde m \leqslant m$, and those can be parallelized $k^{(i)}_{\mathcal Q}$ at a time from Lemma \ref{lemma: main circuit}. From Corollary \ref{Corollary depth}, it follows that the physical depth of any logical circuit on $\mathcal Q^{\otimes m}$ composed of inter-block $\mathsf C^{m-1}\mathsf Z$ gates is bounded by $O\big( (k^{(i)}_{\mathcal Q})^{\,m-1} \big)$.

    \item Second, suppose that $r^{(i)}=\lfloor \eta^{(i)} n^{(i)}\rfloor$, for some constant-approaching term 
    \begin{equation}\label{Equation upper bound constant}
           \eta^{(i)}<\frac{1}{2}\min\{\kappa^{(i)},\delta^{(i)},\delta^{(i)\perp},\beta_q\left(1+\frac{1}{\alpha_q}\right)\},
    \end{equation}
 
    where the $1/2$-factor can be replaced by an arbitrary constant less than one. In that case, we obtain the parameter bounds 
\[
\begin{aligned}
n^{(i)}_{\mathcal Q} & \leqslant \ell^{(i)}\log\ell^{(i)}\left(\frac{1}{\alpha_q}-\frac{\eta^{(i)}}{\beta_q}+1\right)\\
k^{(i)}_{\mathcal Q} &\geqslant \ell^{(i)}\left(\log\ell^{(i)}\frac{\eta^{(i)}}{\beta_q}-1\right)\\
d^{(i)}_{\mathcal Q}   &\geqslant  \ell^{(i)}\left( \frac{ \log \ell^{(i)}(\min\{\delta^{(i)},\delta^{(i)\perp}\}-\eta^{(i)})}{\beta_q} +1\right).
\end{aligned}
\]

In particular, letting $\eta^{(i)}$ approach its upper bound given in Equation \eqref{Equation upper bound constant}, the family $\{\mathcal Q^{(i)}: i\in \NN\}$ has asymptotic parameters scaling as $[[n,\Theta(n), \Theta(n)]]$.
From Theorem \ref{Theorem quasi-G-quantum}, each member $\mathcal Q^{(i)}$ is a free quantum $G^{(i)} $-code which admits an encoder $\Enc^{(i)}$ allowing for the transversal implementation of the logical $\overline{\mathsf C^{\widetilde m-1}\mathsf Z}$ gates acting orbit-wise on tuples of logical qubits for $1< \widetilde m \leqslant m$, and those can be parallelized $k^{(i)}_{\mathcal Q}$ at a time as shown in Lemma \ref{lemma: main circuit}. Moreover, for each member of the family, the gate $\mathsf C^{m}\mathsf Z$ is still transversal, i.e. $(\mathsf C^{m}\mathsf Z)^{\otimes n^{(i)}_{\mathcal Q}}=\overline{\mathsf C^{m}\mathsf Z}^{\otimes k^{(i)}_{\mathcal Q}}$.
\end{enumerate}
In both cases, from Theorem \ref{Theorem decoder}, each code $\mathcal Q^{(i)}$ admits a quasi-quadratic-time decoder for both $X$-type and $Z$-type errors. This is the decoder of the classical AG code $\tilde C^{(i)}$, which has a linear decoding radius and benefits from the fast matrix multiplication in $\F_q[G^{(i)}]$. 
\end{proof}

Here, by transversal gate, we mean a gate acting globally on all the physical qubits across the $m$ blocks of $\mathcal Q$, with an induced action on all the logical qudits: $\mathsf C^{\tilde m-1} \mathsf Z=\overline{\mathsf C^{\tilde m-1} \mathsf Z}^{\otimes k}$.\par

\begin{corollary}[Restatement of Corollary \ref{Corollary MSD}]
Let $q=t^4$ with $t\geqslant 6$. Then, there exists a $[[n,\Theta(n),\Theta(n)]]_q$ quantum code family allowing for a constant overhead distillation protocol of qubit $\mathsf{CCZ}$-magic state with an $O(n^2\cdot \operatorname{polylog}(n))$ time complexity.
\end{corollary}
\begin{proof}
This follows from item \textit{(2)} of Theorem \ref{Theorem two Good quasi-G-quantum}, stating the existence of a $[[n,\Theta(n),\Theta(n)]]_q$ quantum code family with transversal $\mathsf{CCZ}_q$ gate and an $O(n^2\cdot \operatorname{polylog}(n))$-time complexity decoder and a linear decoding radius. We can use this construction in the magic-state distillation protocol of \cite{Wills2025}, converting qudits to qubits using a self-orthogonal basis of $\F_q$ and converting a $\mathsf{CCZ}_q$ magic state into a $\mathsf{CCZ}$ magic state. 
\end{proof}

\newpage
\appendix

\bibliographystyle{alpha}
\bibliography{Biblio_quantum_codes}

@Article{Rosen87,
 Author = {Rosen, Michael},
 Title = {The {Hilbert} class field in function fields},
 FJournal = {Expositiones Mathematicae},
 Journal = {Expo. Math.},
 ISSN = {0723-0869},
 Volume = {5},
 Pages = {365--378},
 Year = {1987},
 Language = {English}
}

@Book{Weil48,
 Author = {Weil, Andr{\'e}},
 Title = {Sur les courbes alg{\'e}briques et les vari{\'e}t{\'e}s qui s'en d{\'e}duisent},
 FSeries = {Publications de l'Institut de Math{\'e}matique de l'Universit{\'e} de Strasbourg},
 Series = {Publ. Inst. Math. Univ. Strasbourg},
 Volume = {7},
 Year = {1948},
 Publisher = {Hermann, Paris},
 Language = {French}
}

@PhdThesis{GasnierPhD,
author = {Gasnier, Jean},
title = {Arithmetics and Algorithmics of Algebraic Curves and Applications to Coding Theory and Cryptography},
school = {Université de Bordeaux},
year = {2025},
note = {Original in French, English version on \url{https://www.math.u-bordeaux.fr/~jgasnier001/}}
}

@inproceedings{CouGas26,
  AUTHOR = {Couveignes, Jean-Marc and Gasnier, Jean},
  TITLE = {Explicit {R}iemann-Roch spaces in the {H}ilbert class field},
  BOOKTITLE = {{A}rithmetic, {G}eometry, {C}ryptography and {C}oding {T}heory},
  EDITOR = {Bruin, Nils and Kohel, David and Martindale, Chloe},
  PUBLISHER = {{AMS}},
  SERIES = {Contemporary {M}athematics},
  VOLUME = {832},
  YEAR = {2026},
  ISBN = {978-1-4704-7685-4},
  DOI = {https://doi.org/10.1090/conm/832}
}

@Book{Liu02,
 Author = {Liu, Qing},
 Title = {Algebraic geometry and arithmetic curves},
 FSeries = {Oxford Graduate Texts in Mathematics},
 Series = {Oxf. Grad. Texts Math.},
 Volume = {6},
 ISBN = {0-19-850284-2},
 Year = {2002},
 Publisher = {Oxford: Oxford University Press},
 Language = {English}
}

@article{BorWil23,
author = {Borello, Martino and Willems, Wolfgang},
title = {On the algebraic structure of quasi-group codes},
journal = {Journal of Algebra and Its Applications},
volume = {22},
number = {10},
pages = {2350222},
year = {2023},
doi = {10.1142/S0219498823502225},
URL = { 
     https://doi.org/10.1142/S0219498823502225
},
eprint = { 
        https://doi.org/10.1142/S0219498823502225 
}
}

@article{BerRioSim09,
author = {Bernal, Jos\'{e} Joaqu\'{\i}n and R\'{\i}o, \'{A}ngel and Sim\'{o}n, Juan Jacobo},
title = {An intrinsical description of group codes},
year = {2009},
publisher = {Kluwer Academic Publishers},
volume = {51},
number = {3},
issn = {0925-1022},
url = {https://doi.org/10.1007/s10623-008-9261-z},
doi = {10.1007/s10623-008-9261-z},
journal = {Des. Codes Cryptography},
month = jun,
pages = {289–300}
}

@article{Shor1995,
  doi = {10.1103/physreva.52.r2493},
  url = {https://doi.org/10.1103/physreva.52.r2493},
  year = {1995},
  month = oct,
  publisher = {American Physical Society ({APS})},
  volume = {52},
  number = {4},
  pages = {R2493--R2496},
  author = {Peter W. Shor},
  title = {{Scheme for reducing decoherence in quantum computer memory}},
  journal = {Physical Review A}
}

@article{Calderbank1996,
  doi = {10.1103/physreva.54.1098},
  url = {https://doi.org/10.1103/physreva.54.1098},
  year = {1996},
  month = aug,
  publisher = {American Physical Society ({APS})},
  volume = {54},
  number = {2},
  pages = {1098--1105},
  author = {A. R. Calderbank and Peter W. Shor},
  title = {{Good quantum error-correcting codes exist}},
  journal = {Physical Review A}
}

@article{Steane1996,
  doi = {10.1103/physrevlett.77.793},
  url = {https://doi.org/10.1103/physrevlett.77.793},
  year = {1996},
  month = jul,
  publisher = {American Physical Society ({APS})},
  volume = {77},
  number = {5},
  pages = {793--797},
  author = {A. M. Steane},
  title = {{Error Correcting Codes in Quantum Theory}},
  journal = {Physical Review Letters}
}

@article{Stean1996Multiple,
  doi = {10.1098/rspa.1996.0136},
  url = {https://doi.org/10.1098/rspa.1996.0136},
  year = {1996},
  month = nov,
  publisher = {The Royal Society},
  volume = {452},
  number = {1954},
  pages = {2551--2577},
  author = {A. M. Steane},
  title = {{Multiple-particle interference and quantum error correction}},
  journal = {Proceedings of the Royal Society of London. Series A: Mathematical,  Physical and Engineering Sciences}
}

@INPROCEEDINGS{8555154,
  author={Fawzi, Omar and Grospellier, Antoine and Leverrier, Anthony},
  booktitle={2018 IEEE 59th Annual Symposium on Foundations of Computer Science (FOCS)}, 
  title={Constant Overhead Quantum Fault-Tolerance with Quantum Expander Codes}, 
  year={2018},
  volume={},
  number={},
  pages={743-754},
  doi={10.1109/FOCS.2018.00076}}

@article{Aharonov97,
author = {Aharonov, Dorit and Ben-Or, Michael},
title = {Fault-Tolerant Quantum Computation with Constant Error Rate},
journal = {SIAM Journal on Computing},
volume = {38},
number = {4},
pages = {1207-1282},
year = {2008},
doi = {10.1137/S0097539799359385},
URL = {  https://doi.org/10.1137/S0097539799359385},
eprint = { https://doi.org/10.1137/S0097539799359385}
}

@article{zhu2024nonclifford,
      title={Non-Clifford and parallelizable fault-tolerant logical gates on constant and almost-constant rate homological quantum LDPC codes via higher symmetries}, 
      author={Guanyu Zhu and Shehryar Sikander and Elia Portnoy and Andrew W. Cross and Benjamin J. Brown},
      year={2024},
      eprint={2310.16982},
      archivePrefix={arXiv},
      primaryClass={quant-ph},
      url={https://arxiv.org/abs/2310.16982}, 
}

@article{zhu2025topological,
      title={A topological theory for qLDPC: non-Clifford gates and magic state fountain on homological product codes with constant rate and beyond the \(N^{1/3}\) distance barrier}, 
      author={Guanyu Zhu},
      year={2025},
      eprint={2501.19375},
      archivePrefix={arXiv},
      primaryClass={quant-ph},
      url={https://arxiv.org/abs/2501.19375}, 
}

@article{Gottesman2013,
author = {Gottesman, Daniel},
title = {Fault-tolerant quantum computation with constant overhead},
year = {2014},
issue_date = {November 2014},
publisher = {Rinton Press, Incorporated},
address = {Paramus, NJ},
volume = {14},
number = {15–16},
issn = {1533-7146},
journal = {Quantum Info. Comput.},
month = nov,
pages = {1338–1372},
numpages = {35}
}

@misc{he2025good,
      title={Asymptotically Good Quantum Codes with Addressable and Transversal Non-Clifford Gates}, 
      author={Zhiyang He and Vinod Vaikuntanathan and Adam Wills and Rachel Yun Zhang},
      year={2025},
      eprint={2507.05392},
      archivePrefix={arXiv},
      primaryClass={quant-ph},
      url={https://arxiv.org/abs/2507.05392}, 
}

@ARTICLE{Stichtenoth2006,
  author={Stichtenoth, H.},
  journal={IEEE Transactions on Information Theory}, 
  title={Transitive and self-dual codes attaining the Tsfasman-Vladut-Zink bound}, 
  year={2006},
  volume={52},
  number={5},
  pages={2218-2224},
  doi={10.1109/TIT.2006.872986}
}

@book{Stichtenothbook,
author = {Stichtenoth, Henning},
title = {Algebraic Function Fields and Codes},
year = {2008},
isbn = {3540768777},
publisher = {Springer Publishing Company, Incorporated},
edition = {2nd},
}

@article{Krishna2019,
  title = {Towards Low Overhead Magic State Distillation},
  author = {Krishna, Anirudh and Tillich, Jean-Pierre},
  journal = {Phys. Rev. Lett.},
  volume = {123},
  issue = {7},
  pages = {070507},
  numpages = {4},
  year = {2019},
  month = {Aug},
  publisher = {American Physical Society},
  doi = {10.1103/PhysRevLett.123.070507},
  url = {https://link.aps.org/doi/10.1103/PhysRevLett.123.070507}
}

@article{PhysRevX.2.041021,
  title = {Magic-State Distillation in All Prime Dimensions Using Quantum Reed-Muller Codes},
  author = {Campbell, Earl T. and Anwar, Hussain and Browne, Dan E.},
  journal = {Phys. Rev. X},
  volume = {2},
  issue = {4},
  pages = {041021},
  numpages = {20},
  year = {2012},
  month = {Dec},
  publisher = {American Physical Society},
  doi = {10.1103/PhysRevX.2.041021},
  url = {https://link.aps.org/doi/10.1103/PhysRevX.2.041021}
}

@article{PhysRevA.86.052329,
  title = {Magic-state distillation with low overhead},
  author = {Bravyi, Sergey and Haah, Jeongwan},
  journal = {Phys. Rev. A},
  volume = {86},
  issue = {5},
  pages = {052329},
  numpages = {10},
  year = {2012},
  month = {Nov},
  publisher = {American Physical Society},
  doi = {10.1103/PhysRevA.86.052329},
  url = {https://link.aps.org/doi/10.1103/PhysRevA.86.052329}
}

@misc{lin2024transversalnoncliffordgatesquantum,
      title={Transversal non-Clifford gates for quantum LDPC codes on sheaves}, 
      author={Ting-Chun Lin},
      year={2024},
      eprint={2410.14631},
      archivePrefix={arXiv},
      primaryClass={quant-ph},
      url={https://arxiv.org/abs/2410.14631}, 
}

@misc{golowich2024quantumldpccodestransversal,
      title={Quantum LDPC Codes with Transversal Non-Clifford Gates via Products of Algebraic Codes}, 
      author={Louis Golowich and Ting-Chun Lin},
      year={2024},
      eprint={2410.14662},
      archivePrefix={arXiv},
      primaryClass={quant-ph},
      url={https://arxiv.org/abs/2410.14662}, 
}

@misc{scruby2024quantumrainbowcodes,
      title={Quantum Rainbow Codes}, 
      author={Thomas R. Scruby and Arthur Pesah and Mark Webster},
      year={2024},
      eprint={2408.13130},
      archivePrefix={arXiv},
      primaryClass={quant-ph},
      url={https://arxiv.org/abs/2408.13130}, 
}

@misc{he2025addressable,
      title={Quantum Codes with Addressable and Transversal Non-Clifford Gates}, 
      author={Zhiyang He and Vinod Vaikuntanathan and Adam Wills and Rachel Yun Zhang},
      year={2025},
      eprint={2502.01864},
      archivePrefix={arXiv},
      primaryClass={quant-ph},
      url={https://arxiv.org/abs/2502.01864}, 
}

@article{PhysRevA.71.022316,
  title = {Universal quantum computation with ideal Clifford gates and noisy ancillas},
  author = {Bravyi, Sergey and Kitaev, Alexei},
  journal = {Phys. Rev. A},
  volume = {71},
  issue = {2},
  pages = {022316},
  numpages = {14},
  year = {2005},
  month = {Feb},
  publisher = {American Physical Society},
  doi = {10.1103/PhysRevA.71.022316},
  url = {https://link.aps.org/doi/10.1103/PhysRevA.71.022316}
}

@misc{nguyen2024goodbinaryquantumcodes,
      title={Good binary quantum codes with transversal CCZ gate}, 
      author={Quynh T. Nguyen},
      year={2024},
      eprint={2408.10140},
      archivePrefix={arXiv},
      primaryClass={quant-ph},
      url={https://arxiv.org/abs/2408.10140}, 
}

@misc{golowich2024asymptoticallygoodquantumcodes,
      title={Asymptotically Good Quantum Codes with Transversal Non-Clifford Gates}, 
      author={Louis Golowich and Venkatesan Guruswami},
      year={2024},
      eprint={2408.09254},
      archivePrefix={arXiv},
      primaryClass={quant-ph},
      url={https://arxiv.org/abs/2408.09254}, 
}

@article{Wills2025,
  title = {Constant-overhead magic state distillation},
  ISSN = {1745-2481},
  url = {http://dx.doi.org/10.1038/s41567-025-03026-0},
  DOI = {10.1038/s41567-025-03026-0},
  journal = {Nature Physics},
  publisher = {Springer Science and Business Media LLC},
  author = {Wills,  Adam and Hsieh,  Min-Hsiu and Yamasaki,  Hayata},
  year = {2025},
  month = sep 
}

@article{PhysRevLett.102.110502,
  title = {Restrictions on Transversal Encoded Quantum Gate Sets},
  author = {Eastin, Bryan and Knill, Emanuel},
  journal = {Phys. Rev. Lett.},
  volume = {102},
  issue = {11},
  pages = {110502},
  numpages = {4},
  year = {2009},
  month = {Mar},
  publisher = {American Physical Society},
  doi = {10.1103/PhysRevLett.102.110502},
  url = {https://link.aps.org/doi/10.1103/PhysRevLett.102.110502}
}

@article{PhysRevA.91.032330,
  title = {Universal transversal gates with color codes: A simplified approach},
  author = {Kubica, Aleksander and Beverland, Michael E.},
  journal = {Phys. Rev. A},
  volume = {91},
  issue = {3},
  pages = {032330},
  numpages = {12},
  year = {2015},
  month = {Mar},
  publisher = {American Physical Society},
  doi = {10.1103/PhysRevA.91.032330},
  url = {https://link.aps.org/doi/10.1103/PhysRevA.91.032330}
}

@article{Kubica_2015,
doi = {10.1088/1367-2630/17/8/083026},
url = {https://doi.org/10.1088/1367-2630/17/8/083026},
year = {2015},
month = {aug},
publisher = {IOP Publishing},
volume = {17},
number = {8},
pages = {083026},
author = {Kubica, Aleksander and Yoshida, Beni and Pastawski, Fernando},
title = {Unfolding the color code},
journal = {New Journal of Physics}
}

@article{PhysRevLett.111.090505,
  title = {Universal Fault-Tolerant Quantum Computation with Only Transversal Gates and Error Correction},
  author = {Paetznick, Adam and Reichardt, Ben W.},
  journal = {Phys. Rev. Lett.},
  volume = {111},
  issue = {9},
  pages = {090505},
  numpages = {5},
  year = {2013},
  month = {Aug},
  publisher = {American Physical Society},
  doi = {10.1103/PhysRevLett.111.090505},
  url = {https://link.aps.org/doi/10.1103/PhysRevLett.111.090505}
}

@article{PhysRevA.100.012312,
  title = {Three-dimensional surface codes: Transversal gates and fault-tolerant architectures},
  author = {Vasmer, Michael and Browne, Dan E.},
  journal = {Phys. Rev. A},
  volume = {100},
  issue = {1},
  pages = {012312},
  numpages = {20},
  year = {2019},
  month = {Jul},
  publisher = {American Physical Society},
  doi = {10.1103/PhysRevA.100.012312},
  url = {https://link.aps.org/doi/10.1103/PhysRevA.100.012312}
}

@article{PhysRevLett.97.180501,
  title = {Topological Quantum Distillation},
  author = {Bombin, H. and Martin-Delgado, M. A.},
  journal = {Phys. Rev. Lett.},
  volume = {97},
  issue = {18},
  pages = {180501},
  numpages = {4},
  year = {2006},
  month = {Oct},
  publisher = {American Physical Society},
  doi = {10.1103/PhysRevLett.97.180501},
  url = {https://link.aps.org/doi/10.1103/PhysRevLett.97.180501}
}

@article{PhysRevA.95.012329,
  title = {Diagonal gates in the Clifford hierarchy},
  author = {Cui, Shawn X. and Gottesman, Daniel and Krishna, Anirudh},
  journal = {Phys. Rev. A},
  volume = {95},
  issue = {1},
  pages = {012329},
  numpages = {7},
  year = {2017},
  month = {Jan},
  publisher = {American Physical Society},
  doi = {10.1103/PhysRevA.95.012329},
  url = {https://link.aps.org/doi/10.1103/PhysRevA.95.012329}
}

@article{Bombin_2015,
doi = {10.1088/1367-2630/17/8/083002},
url = {https://doi.org/10.1088/1367-2630/17/8/083002},
year = {2015},
month = {aug},
publisher = {IOP Publishing},
volume = {17},
number = {8},
pages = {083002},
author = {Bombín, Héctor},
title = {Gauge color codes: optimal transversal gates and gauge fixing in topological stabilizer codes},
journal = {New Journal of Physics}
}

@misc{guemard2025good,
      title={Good quantum codes with addressable and parallelizable transversal non-Clifford gates}, 
      author={Virgile Gu{\'{e}}mard},
      year={2025},
      eprint={2510.19809},
      archivePrefix={arXiv},
      primaryClass={quant-ph},
      url={https://arxiv.org/abs/2510.19809}, 
}

@book{Huffman_Pless_2003, place={Cambridge}, title={Fundamentals of Error-Correcting Codes}, publisher={Cambridge University Press}, author={Huffman, W. Cary and Pless, Vera}, year={2003}}

@misc{li2026transversal,
      title={Transversal non-Clifford gates on almost-good quantum LDPC and quantum locally testable codes}, 
      author={Yiming Li and Zimu Li and Zi-Wen Liu},
      year={2026},
      eprint={2604.01874},
      archivePrefix={arXiv},
      primaryClass={quant-ph},
      url={https://arxiv.org/abs/2604.01874}, 
}
\end{document}